\documentclass[sigconf]{acmart}

\AtBeginDocument{%
  \providecommand\BibTeX{{%
    \normalfont B\kern-0.5em{\scshape i\kern-0.25em b}\kern-0.8em\TeX}}}

\acmYear{2020}\copyrightyear{2020}
\setcopyright{acmcopyright}
\acmConference[FAT* '20]{Conference on Fairness, Accountability, and Transparency}{January 27--30, 2020}{Barcelona, Spain}
\acmBooktitle{Conference on Fairness, Accountability, and Transparency (FAT* '20), January 27--30, 2020, Barcelona, Spain}
\acmPrice{15.00}
\acmDOI{10.1145/3351095.3372847}
\acmISBN{978-1-4503-6936-7/20/01}

\usepackage{amsmath}               
\usepackage{amsthm}                
\usepackage{amssymb}               
\usepackage{fancyhdr}              
\usepackage[shortlabels]{enumitem}
\PassOptionsToPackage{hyphens}{url}\usepackage{hyperref}
\usepackage[ruled, linesnumbered]{algorithm2e}
\usepackage[autostyle=true,german=quotes]{csquotes}%

\usepackage{bm}
 
\usepackage{graphicx}
\usepackage{color}

\newcommand\citein[1]{\citeauthor{#1} (\citeyear{#1})}

\setlist[enumerate]{leftmargin=*}  
\setlist[itemize]{leftmargin=*}        

\newtheorem{theorem}{Theorem}
\newtheorem*{theorem*}{Theorem}
\newtheorem{lemma}{Lemma}
\newtheorem*{lemma*}{Lemma}

\newcommand{\p}[1]{\left( #1 \right)}
\newcommand{\br}[1]{\left[ #1 \right]}
\newcommand{\ceil}[1]{\left\lceil #1 \right\rceil}
\newcommand{\floor}[1]{\left\lfloor #1 \right\rfloor}

\newcommand{\ngroup}[0]{N}
\newcommand{\res}[0]{v}
\newcommand{\bigres}[0]{V}
\newcommand{\resdist}[0]{\mathcal{V}}
\newcommand{\can}[0]{c}
\newcommand{\bigcan}[0]{C}
\newcommand{\candist}[0]{\mathcal{C}}
\newcommand{\fair}[0]{\alpha}
\newcommand{\pof}[0]{\rho}
\newcommand{\budget}[0]{R}
\newcommand{\npower}[0]{h}
\newcommand{\powexp}[0]{b}

\newcommand{\cd}[0]{\cdot}

\usepackage{array}
\newcolumntype{C}[1]{>{\centering\let\newline\\\arraybackslash\hspace{0pt}}m{#1}}

\begin{document}

\title[Fairness \&  Utilization in Allocating Resources]{Fairness and Utilization in Allocating Resources \\ with Uncertain Demand}

\author{Kate Donahue}
\email{kdonahue@cs.cornell.edu}
\affiliation{%
  \institution{Department of Computer Science, Cornell University}
}

\author{Jon Kleinberg}
\affiliation{Departments of Computer Science and Information Science, Cornell University}
\email{kleinber@cs.cornell.edu}

\renewcommand{\shortauthors}{Donahue and Kleinberg}

\begin{abstract}
Resource allocation problems are a fundamental domain in which to evaluate the fairness properties of algorithms. The trade-offs between fairness and utilization have a long history in this domain. A recent line of work has considered fairness questions for resource allocation when the demands for the resource are distributed across multiple groups and drawn from probability distributions. In such cases, a natural fairness requirement is that individuals from different groups should have (approximately) equal probabilities of receiving the resource. A largely open question in this area has been to bound the gap between the maximum possible utilization of the resource and the maximum possible utilization subject to this fairness condition. 

Here, we obtain some of the first provable upper bounds on this gap. We obtain an upper bound for arbitrary distributions, as well as much stronger upper bounds for specific families of distributions that are typically used to model levels of demand. In particular, we find --- somewhat surprisingly --- that there are natural families of distributions (including Exponential and Weibull) for which the gap is non-existent: it is possible to simultaneously achieve maximum utilization and the given notion of fairness. Finally, we show that for power-law distributions, there is a non-trivial gap between the solutions, but this gap can be bounded by a constant factor independent of the parameters of the distribution.
\end{abstract}




\maketitle

\section{Introduction}
The current interest in fairness properties of algorithms includes several distinct themes, one of which is the question of fairness in allocating scarce resources. Research on this question has a long history, with foundational work in the 1980s and 1990s on allocating resources in computer systems \cite{Jaffe81bottleneckflow,Chiu1989analysis,kelly1998rate, Demers:1989:ASF:75247.75248, aimd}.  Allocation problems continue to form an important topic for fairness considerations, especially as automated systems make allocation decisions in a wide range of areas that reach far beyond the original computational settings of the problem \cite{eubanks2018automating,procaccia13cakecacm}.

Recently, \citein{Elzayn_2019} considered a novel allocation problem in this style. In their formulation, individuals are divided into groups, each of which has some probability distribution of \emph{candidates} who desire the resource. There is a fixed amount of resource $\budget$ and the question is how to allocate this resource across the groups. Such a set-up arises in many applications; as one concrete motivating example for purposes of discussion, suppose that we have a set number of doctors, each with a maximum number of people they can assist, and we would like to allocate them across a set of $\ngroup$ geographically distributed regions. Each region contains a population with a potentially different probability distribution over the number of sick people in need of doctors.  We know the distribution of demand for each region, but we will only know the actual demand --- drawn from these distributions --- once we have allocated the doctors by dividing them up in some fashion across the regions.  If a community has too few doctors, some sick people will go unassisted; and if a community has more doctors than it needs for the number of sick people it has, some of the doctor's capacity to help will go unused. Formally, the expected number of candidates that receive a resource under a certain allocation is called the \emph{utilization} of that allocation. 

One allocation strategy could be to put the doctors where they are most likely to be used, with the aim of maximizing utilization. However, such an allocation could potentially \enquote{starve} certain regions, leaving them with very few doctors even though they have a non-trivial level of need. Such an allocation violates a natural definition of fairness: sick people in some regions would have a higher probability of receiving assistance than sick people in other regions, which conflicts with the premise that people are equally deserving of care from the doctors. 

In this work, we thus say that an allocation is \emph{fair} across multiple regions if a candidate who needs the resource has the same probability of receiving it, regardless of the region they belong to: the identity of the region doesn't impact the probability that those in need obtain assistance. At some points, we relax this definition slightly to say that an allocation is $\fair$-\emph{fair} if the probability that a candidate receives the resource in different regions is within a maximum difference of $\fair$, for $\fair \in [0, 1]$ across all of the regions. Note that $\fair=0$ corresponds to our initial notion of fairness, and $\fair=1$ imposes no constraint at all on the allocation.

\citein{Elzayn_2019} consider a version of this allocation problem in the scenario where the candidate probability distributions are unknown and must be learned. At each time step, a certain allocation is chosen and the feedback obtained reveals only the number of candidates who received the resource, not the true number. They adapt learning results from previous dark pool trading scenarios proposed in \citein{ganchev2009censored}. From these results, they construct algorithms that learn utilization-maximizing and $\fair$-fair utilization maximizing allocations through this type of censored feedback\footnote{They consider a related, but different definition of fairness from the one used in this work. Later sections will expand on this difference.}. They also define the \emph{Price of Fairness} (PoF) as the ratio of the utilizations of the max-utilizing allocation and the max-utilizing $\fair$-fair allocation. From an empirical dataset, they calculate the PoF for various levels of $\fair$ in practice. However, they leave largely open the following category of questions: can we obtain theoretical upper bounds on the price of fairness over different probability distributions?  These questions appear to become quite complex even when the distributions are known.

\paragraph{\bf The present work: The interaction of fairness and utilization.}
In this work we obtain bounds on the price of fairness for this family of resource allocation problems, both through general results that hold for all distributions and through stronger results that are specific to common families of distributions.  We also distinguish between two versions of the problem --- one in which we require that the allocation to each group be integer-valued, and one in which we allow resources to be allocated fractionally (or probabilistically, which yields equivalent results in expectation).

In the case of integer-valued resource allocation, we show through a constructive proof that the PoF can be unboundedly large. If we allow resources to be allocated fractionally, then we show that the PoF is bounded above by $\frac{1}{\fair}$ for arbitrary distributions and $\fair > 0$. In the case where $\fair=0$, the PoF can  be unboundedly large for fractional allocations as well.

We then show that much stronger upper bounds can be obtained for large families of  distributions that are often used to model levels of demand.  First we show that  certain families of natural distributions have PoF equal to 1, the smallest possible bound: for these distributions, there is no trade-off between fairness and maximum utilization. We show that distributions with this property include Exponential and Weibull distributions. We then consider the family of Power Law distributions for a constant number of groups; we show that these distributions have a PoF that can be strictly greater than $1$, but is always bounded above by a universal constant. Table \ref{tab:overview} in Appendix \ref{table} contains a high-level summary of results explored in this paper. 

Overall, our results reveal a rich picture in the trade-off between fairness and utilization for this type of allocation problem, through the different families of bounds for the price of fairness --- ranging from distributions where perfect fairness can be achieved while maximizing utilization, to those in which the gap is bounded by a constant, to others for which the gap can be large.  The techniques used for the analysis suggest further opportunities for reasoning about how balancing the probability of service to different groups can both constrain and also be compatible with other objectives. 

\paragraph{\bf Important limitations.} The assumptions and limitations of this model are discussed in greater detail in Section \ref{model}. However, it is worth emphasizing some of the specific choices inherent in creating this model. In particular, the notion of fairness that this paper uses is \enquote{equality in the chance of obtaining a resource, conditional on needing it}. Although we describe in the next section why this is a natural definition of fairness to use, it is certainly not the only possible definition of fairness, and alternative definitions could lead to different results. For example, a notion of fairness could take into account historical discrepancies in resources between locations and attempt to proactively rectify them. 

Another fundamental assumption of this work is that there is a fixed amount of resources available to distribute. An alternative viewpoint could be to explore ways to increase the total amount of resources so that more people in need are able to obtain help. Increasing available resources is a critical goal, and one where advocacy and political action can play a key role in bringing this about for real-world allocation problems. However, in almost all cases it is only feasible to reduce, not fully eliminate, unmet need. Thus, even after the overall level of resources have been increased, fair resource allocation will still be a problem that needs to be addressed. In this way, we view both increasing the total amount of resources and fair resource allocation as important and complementary areas of study; in most cases, neither eliminates the need for the other.

\section{Motivating example}\label{motivate}
To motivate the problem and to make some of the types of calculations more clear, we start with the following example. 
Suppose there is a remote stretch of coastline with two small villages, $A$ and $B$, each with a small number of houses. This stretch of coastline is prone to severe localized storms that lead to power outages. Each village has a slightly different probability distribution over storms - distributions that are known. In any particular week, the probability distribution over the number of houses $\bigcan$ impacted by power outages in each village is as follows: 
$$P_A(\bigcan = \can) =  \begin{cases} 
     0.6 & \can = 0 \\
     0.4 & \can = 2 \\
     0 & \text{otherwise} 
   \end{cases}
\quad 
P_B(\bigcan = \can) =  \begin{cases} 
     0.3 & \can = 0 \\
     0.7 & \can = 3 \\
     0 & \text{otherwise}
   \end{cases}
$$
In our example, assume that each week's power outage is independent of other weeks, and the villages are far enough apart that the power outages in one village aren't correlated with the power outages in the other. A regional planning committee has at its disposal 2 generators, each of which can restore power to 1 house. They are trying to decide how to allocate these generators across the two villages. The generators cannot be transferred from one village to another after a storm strikes. If there are $\res$ generators allocated to a village and $\can$ houses in need of generators, the number of houses that receive generators in that village is $\min(\res, \can)$. 

We include this example only as a motivating case for our broader model: we do \emph{not} intend that this model apply exclusively to disaster relief allocation. However, we do wish to make it clear that maximizing utilization and fairness in disaster relief resource allocation is an active area of study \cite{doi:10.1111/risa.13342}.

The regional planning committee first decides that it wishes to maximize \emph{utilization}: the expected number of houses in need of generators that receive them. 
There are three options for allocations: $\{(\res_A = 0, \res_B = 2), (\res_A = 1, \res_B = 1), (\res_A = 2, \res_B = 0)\} $ where $\res_i = j$ means that village $i$ gets $j$ generators.  Formally, utilization can be written as: 
$$U(\res_A = n_A, \res_B = n_B) = \mathbb{E}_{\bigcan \sim \candist_A}\br{\text{min}(\bigcan, n_A)} + \mathbb{E}_{\bigcan \sim \candist_B}\br{\text{min}(\bigcan, n_B)} $$ 
The utilizations across these three allocations are given in Table \ref{tab:ex_util}. The allocation that maximizes utilization is $(\res_A = 0, \res_B = 2)$. 
\begin{table}[]
\centering
\begin{tabular}{|C{1cm}|C{2.5cm}|C{2.5cm}|C{1.2cm}|}
\hline
$\{\res_A,\res_B\}$ & $\mathbb{E}_{\bigcan \sim \candist_A}\br{\text{min}(\bigcan, \res_A)}$ & $\mathbb{E}_{\bigcan \sim \candist_B}\br{\text{min}(\bigcan, \res_B)}$ & $U(\res_A, \res_B)$ \\ \hline
$\{0, 2\}$ & $0.6 \  \text{min}(0, 0) + 0.4 \  \text{min}(2, 0) = 0$                & 

$0.3\ \text{min}(0, 2) + 0.7\ \text{min}(3, 2) = 1.4$                  & $1.4$           \\\hline 
$\{1, 1\}$ & $0.6 \ \text{min}(0, 1) + 0.4 \ \text{min}(2, 1) = 0.4$                & $0.3 \ \text{min}(0, 1) + 0.7 \ \text{min}(3, 1) = 0.7$                & $1.1$               \\ \hline
$\{2, 0\}$ & $0.6\ \text{min}(0, 2) +0.4\ \text{min}(2, 2) = 0.8$                   & $0.3 \ \text{min}(0, 0) + 0.7\ \text{min}(3, 0) = 0$                   & $0.8$               \\ \hline
\end{tabular}
\caption{Utilization of resources across two groups in motivating example.}
\label{tab:ex_util}
\vspace{-6mm}
\end{table}

However, something bothers the regional committee: it feels fundamentally unfair that none of the houses in village A ever receive generators. To formalize this, we could ask what fraction of houses in need of generators obtain them, on average, and aim to select an allocation that brings this proportion between the two villages as close to equality as possible. First, we derive the relevant probability for an arbitrary village:

\begin{eqnarray*}
& & P\p{\text{house $k$ gets generator} \mid \text{house $k$ needs generator}} \\
&= &\frac{P\p{\text{house $k$ gets generator }\cap \text{ house $k$ needs generator}}}{P(\text{house $k$ needs generator})} \\ 
&= &  \frac{P(\text{house $k$ gets generator })}{P(\text{house $k$ needs generator})}
\end{eqnarray*}
We assume that there is a total number of houses $0<K<\infty$ in the village and that all houses within the same village are interchangable. The total number of houses in the village that need generators is given by $\bigcan \sim \candist$. Given that $\can$ houses need generators, there is a $\frac{\can}{K}$ probability that a randomly selected house $k$ will be in need. Given that a house is in need, the probability it obtains a generator equals $\min(\frac{\res}{\can}, 1)$. $P(\text{house $k$ needs generator})$ equals
\begin{eqnarray*}
& = & \sum_{\can = 0}^{K}P(\text{house $k$ needs generator} \mid \ \bigcan = \can ) \cd P(\bigcan = \can) \\ 
& =  & \sum_{\can=0}^K\frac{\can}{K}P(\bigcan = \can) = \sum_{\can=1}^K\frac{\can}{K}P(\bigcan = \can)
\end{eqnarray*}

Similarly, we can use the $K$ term to expand  $P(\text{house $k$ gets generator})$:
\begin{eqnarray*}
&=& P(\text{$k$ gets generator} \mid \text{$k$ needs generator})P(\text{$k$ needs generator })\\ 
&& + \  P(\text{$k$ gets generator} \mid \text{$k$ does NOT need generator}) \cd  \\&& P(\text{$k$ does NOT need generator })
\end{eqnarray*}
$P(\text{house $k$ gets generator} \mid \text{house $k$ does NOT need generator}) = 0$. We expand the remaining term by conditioning on $\bigcan$. The sum starts from $\can = 1$ because if $\can=0$ there are 0 houses in need of generators. 
\begin{eqnarray*}
\sum_{\can=1}^K && P(\text{$k$ gets generator} \vert \text{$k$ needs generator},\  \bigcan = \can) \cd  \\ && P(\text{$k$ needs generator } \vert \bigcan = \can)P(\bigcan = \can) \\ 
 &=& \sum_{\can=1}^K\min\p{\frac{\res}{\can}, 1}\cd \frac{\can}{K}\cd P(\bigcan = \can)
\end{eqnarray*}

Combining these results gives: 

\begin{eqnarray*}
\frac{P(\text{house $k$ gets generator })}{P(\text{house $k$ needs generator})} & = & \frac{\sum_{\can=1}^K\min\p{\frac{\res}{\can}, 1}\cd \frac{\can}{K}\cd P(\bigcan = \can)}{\sum_{\can=1}^K\frac{\can}{K}\cd P(\bigcan = \can)} \\
& = & \frac{\sum_{\can=1}^K\min\p{\frac{\res}{\can}, 1}\cd \can \cd P(\bigcan = \can)}{\sum_{\can=1}^K\can \cd P(\bigcan = \can)} \\
& = & \frac{\sum_{\can=1}^K\min(\res, \can)P(\bigcan = \can)}{\sum_{\can=1}^K \can \cd P(\bigcan = \can)} \\
& = & \frac{\mathbb{E}_{\bigcan \sim \candist}\min(\res,\bigcan)}{\mathbb{E}_{\bigcan \sim \candist}\bigcan} \\
& =: & q(\res_i, \candist_i)
\end{eqnarray*}

where in the last step we have defined the probability of receiving the resource, conditional on needing it, as $q(\res_i, \candist_i)$, where $i$ indexes the village. Note that this definition differs from \citein{Elzayn_2019}: that definition calculated the probability of candidates receiving the resource weighted by each time period. By contrast, our definition is based on the per-person probability of receiving the resource, which we believe is a more natural definition. 

One goal might be to bring these probabilities across the villages to be as close as possible to each other. How does each allocation considered in Table \ref{tab:ex_util} do under this definition? For reference, $\mathbb{E}_{\bigcan \sim \candist_A}\bigcan = 0.6 \cd 0 + 0.4 \cd 2 = 0.8$ and $\mathbb{E}_{\bigcan \sim \candist_B}\bigcan = 0.3 \cd 0 + 0.7\cd 3 = 2.1$. The results are shown in Table \ref{tab:ex_fair}. 

\begin{table}[]
\centering
\begin{tabular}{|C{1.1cm}|C{1.3cm}|C{1.3cm}|C{1.6cm}|C{1.2cm}|}
\hline
$\{\res_A, \res_B\}$& $q(\res_A, \candist_A)$ & $q(\res_B, \candist_B)$        & $\vert q(\res_A, \candist_A) - q(\res_B, \candist_B) \vert $ & $U(\res_A, \res_B)$\\[3pt] \hline
$\{0, 2\}$ & $\frac{0}{0.8} = 0$     & $\frac{1.4}{2.1} = 0.\overline{6}$ & $0.\overline{6}$ & 1.4                                        \\[3pt] \hline
$\{1, 1\}$ & $\frac{0.4}{0.8} = 0.5$ & $\frac{0.7}{2.1} = 0.\overline{3}$ & $0.1\overline{6}$ & 1.1                                         \\[3pt] \hline
$\{2,0\}$ & $\frac{0.8}{0.8} = 1$   & $\frac{0}{2.1} = 0$                & $1$    & 0.8                                                   \\[3pt] \hline
\end{tabular}
\caption{Probability of receiving resources across two groups in motivating example. Utilization included for reference.}
\label{tab:ex_fair}
\vspace{-10mm}
\end{table}

Given these numbers, the regional committee could look at the utilization and fairness of each allocation and make their own decision about how resources should be allocated. In particular, they might decide that they want the $q(\res_i, \candist_i)$ fractions between each village to be within a certain difference $\fair$ and pick the allocation with maximal utilization subject to those constraints. It is this scenario that much of this paper will consider: comparing max-utilization with max-utilization subject to an $\fair$-fairness constraint. 

\section{Applications and formal model}\label{model}

Our model is best described as a two-stage stochastic problem in which resource allocation decisions are made in the first stage, and demand is drawn probabilistically in the second stage. We envision our work being used in the following way. Suppose that the group identity is associated with some decomposition of the population into groups (whether by geographic location, or socioeconomic status, or attributes like race or ethnicity), such that we wish for resources to be fairly distributed across these groups. In this case, the PoF bound could be used to provide a general sense of the maximum trade-off between utilization and fairness that a decision-maker might need to face. The purpose of this paper is not to provide guidance for which allocation decision-makers should choose, but rather to provide general bounds on when this procedure will result in a large trade-off of utilization and fairness, and when the trade-off will be smaller. 

Much of this paper will focus on cases with many candidates and many resources, which in the limit will be be approximated by continuous probability distributions. This approximation is common: in many application areas in social sciences, it is known that certain phenomena are well-described by power law and exponential distributions, even when their exact parameters might be unclear \cite{power_law_emp} \cite{Clauseteaao3580}. We have focused much of our analysis on these distributions with the goal that our results can be maximally useful. 

\subsection{Model detail}
In this section, we make the motivating example more abstract. We assume a set of $\ngroup$ groups. Each group $i$ has a distribution $\candist_i$ over the number of candidates $\bigcan$ that are present during a particular unit of time.  We assume that 
$$\mathbb{E}_{\bigcan \sim \candist_i}\bigcan > 0 \ \forall i\in [\ngroup] \text{ and } \candist_i \perp\!\!\!\perp \candist_j \quad \forall i \ne j$$
and that each time period's number of candidates is independently and identically distributed. Much of this paper will focus on the continuous probability distribution case, but the results obtained largely do not depend on this factor. 

In the model, we assume we have a number $\budget \in \mathbb{N}$ of \emph{resources} that we can allocate across these groups. We assume that allocations must be selected before the number of candidates is realized each time period. If $\res$ resources and $\can$ candidates are present a given group, then $\min(\res, \can)$ candidates will receive the resource. If $\res > \can$, some resources will go unused, and if $\res < \can$, some candidates will go without resources. We assume that the only distinguishing characteristic candidates have is which group they come from: otherwise, they are interchangeable. 

We have two different objectives: \emph{utilization} and \emph{fairness}.
\begin{eqnarray*}
\text{Utilization: }U(\textbf{\res}, \{\candist_i\}) \coloneqq&& \sum_{i=1}^{\ngroup} \mathbb{E}_{\bigcan \sim \candist_i} \min(\bigcan, \res_i) \\
\text{Fairness: } Q(\textbf{\res}, \{\candist_i\}) \coloneqq&& \max_{i,j} \left \vert q(\res_i, \candist_i) - q(\res_j, \candist_j) \right\vert   \\ 
  && \text{ with } q(\res, \candist) \coloneqq \frac{\mathbb{E}_{\bigcan \sim \candist}\min(\bigcan, \res)}{\mathbb{E}_{\bigcan \sim \candist}\bigcan}
\end{eqnarray*}
The fairness term bounds the difference between groups in the probability of a candidate receiving a resource. An allocation ${\bf{\res}}$ is $\fair$-fair if the fairness objective has value $\leq \fair$, so all $q(\res_i, \candist_i)$ fractions are within $\fair$ of each other. Next, we define the Price of Fairness:
\begin{eqnarray*}
\text{PoF}(\fair) & =& 
 \frac{\text{Unconstrained utilization}}{\text{Maximum utilization that is $\fair$-fair}} \\ 
& = & \frac{\max_{\textbf{\res}} U( \textbf{\res}, \{\candist_i\}) \text{ s.t } \sum_i \res_i = \budget}{\max_{\textbf{\res}} U(\textbf{\res}, \{\candist_i\}) \text{ s.t } Q(\textbf{\res}, \{\candist_i\}) \leq \fair \text{ and } \sum_i \res_i = \budget} 
\end{eqnarray*}

If there exists an allocation that simultaneously maximizes utilization and is $\fair =0$ fair, then the PoF is 1 and there is no tension between our two objectives. Otherwise, the PoF is $>1$ and we will need to trade off between these objectives in order to make a decision about which allocation to choose. We will use the standard notation of having $f(x)$ represent the probability density function (or probability mass function) of a distribution and $F(x)$ represent the cumulative distribution function. 

\paragraph{\bf Assumptions.}
There are a few assumptions in this model that are important to recognize. To start, we require that any allocation always use its entire budget $\budget$. It might be possible that an $\fair$-fair allocation might prefer to use less than its entire budget in order to achieve fairness: this is a case that our model specifically disallows. Additionally, it would often be the case that increasing the budget $\budget$ would make it easier to achieve any given fairness metric. We prove in Lemma \ref{increase_res} that increasing the budget only improves fairness and utilization metrics, but overall the area of increasing or decreasing the available resources is not the focus of this paper.

Secondly, our definition of fairness revolves around the maximum difference in probability of receiving the resource. A different metric could have used the average difference in probability, for example. A more ambitious multi-time period extension could take into account historical differences in resource allocation and have a fairness metric that attempts to rectify past inequalities. While these (and other) different models might also be valid and interesting future avenues to explore, they would all be built on the necessary foundation this paper provides. The rest of this paper will use the assumptions and model listed above.

\subsection{Application areas}
Despite the motivating example in Section \ref{motivate}, the focus of this paper is not to model any one specific application such as disaster relief aid. Instead, we view this model as representing the core features of allocation problems in a variety of contexts. A few of these are listed below. 
\begin{itemize}
    \item \textbf{Doctors}: As described in the introduction, one example could be allocating doctors across regions with different probabilities of ill patients. It is worth noting that the presence of doctors would likely influence the level of illnesses. 
    \item \textbf{Blood banks}: Donated blood needs to be allocated across different hospitals with potentially different needs for blood. This area has the added complication that blood donations are not interchangeable: candidates can only receive a certain subset of the available types of blood. 
    \item \textbf{Schools}: A public school district might try to allocate certain resources (teachers, computers, specialized classes) across schools of different size, and therefore different distributions of need. 
\end{itemize}
We have noted that in each case, the potential application area would have key features that are not already captured in the above model. This is intentional because the goal of this paper is to provide broad results for a broad class of problems. Both \citein{Elzayn_2019} and \citein{ensign18a} use related models in the contexts of allocating police officers to districts with differing crime rates. While this is also an area of application for this model, it raises additional issues that we do not model here, such as the likelihood that the presence of police officers will have an impact on the rate of crime. 

Fairness in resource allocation work has been studied in various application areas. For example, \cite{Jaffe81bottleneckflow,Chiu1989analysis,kelly1998rate, Demers:1989:ASF:75247.75248} study the issue of fair allocation of bandwidth, speed of service, and buffer space in networking. These works differ from ours in the selected model: it is assumed that actors can modify the amount of resources requested at will, which differs from our probabilistic resource generation method. \citein{doi:10.1111/risa.13342} studies the specific case of disaster relief aid. In this case, it is assumed that the amount of need is determined before the allocation decision needs to be made, which again differs from our model.

\section{General bounds on Price of Fairness}\label{generalbounds}
One natural question we might ask is, \enquote{What is the maximum amount of utilization we give up by requiring fairness?} This amount depends on the set of demand distributions $\{\candist_i\}$ and our budget $\budget$. If we have these quantities, we can directly calculate the PoF by first calculating the unconstrained max-utilization, then the $\fair$-fair max-utilization. \citein{Elzayn_2019} provides algorithms to calculate both of these quantities (so long as the allocations are integers). However, in many cases we might not know the specific $\{\candist_i\}$ in advance, but we might still want to have a bound on the PoF. In this section, we explore general bounds on the PoF when we do not know the candidate distributions in advance. 

\subsection[Discrete allocation]{Discrete resource allocation has unbounded PoF}
Suppose we allocate resources in integer units: for example, one generator at a time. In this formulation, we find that \emph{the PoF is unbounded: achieving $\fair$ fairness could require giving up arbitrarily large amounts of utilization}. This proof is presented in detail in Appendix \ref{supp}. A rough sketch is below. 

Suppose we are given a desired $\fair < 1$ and a maximum PoF $\pof$ we are willing to accept. (We require $\fair <1$ because if $\fair = 1$ the fairness constraint has no effect, so the PoF will be 1 always.) Then, it is possible to construct a set of candidate distributions $\{\candist_i\}$ and a set of resources $\budget$ over this such that the PoF is $> \pof$. The proof involves creating two different candidate distributions with parameters depending on $\fair$ and $\pof$ and then showing that the PoF of allocating resources across those groups is always $>\pof$.

\subsection[Fractional allocation $1/\fair$]{PoF bounded by $1/\fair$ under fractional allocation}
A critical reader might wonder about the PoF if we were allowed to allocate resources fractionally. Such allocations would offer more flexibility and might allow for lower PoF. In the motivating example of allocating generators across villages, for example, this might correspond to allocating 0.7 generators to Village A and 1.3 generators to Village B. For divisible resources (such as potable water stores), this type of allocation makes intuitive sense. For indivisible resources, like generators, we can view the allocation as \emph{probabilistic}: An allocation of $\res = 1.3$, for example, could be viewed as a 70\% chance of $\res= 1$ and a 30\% chance of $\res = 2$. 

The expected utilization is equal to $\mathbb{E}_{\bigres \sim \resdist}\mathbb{E}_{\bigcan \sim \candist}\text{min}(\bigres, \bigcan)$. Both the \enquote{deterministic, fractional} and \enquote{probabilistic, integer-valued} interpretations give in the same expected utilization (proven in Appendix \ref{supp}). We will switch between these two interpretations (probabilistic and deterministic) depending on which makes more sense for the problem at hand. 

Given the ability to allocate resources fractionally, we will find that the PoF is no longer unbounded. In particular, the PoF is upper bounded by $\frac{1}{\fair}$. Note that this proof relies on an algorithm to calculate a max-utilizing allocation in the fractional case, as well as other related proofs, all present in Appendix \ref{algorithms}. 

\begin{theorem}
Given a set of candidate distributions $\{\candist_i\}$, level of resources $\budget$, and ability to allocate resources fractionally, it is possible to find an $\fair$-fair allocation with PoF at most $\frac{1}{\fair}$, though this allocation may not use all of the resources.
\end{theorem}
\begin{proof}
First, we find a max-utilization allocation and  divide the groups into two categories:
$$A = \{i \ \vert \ q(\res_i, \candist_i) > \fair\} \quad B = \{i \ \vert \ q(\res_i , \candist_i) \leq \fair\}$$
If $A = \varnothing$, then the max-utilization allocation is already $\fair$-fair and PoF = $1 < \frac{1}{\fair}$. Otherwise, modify the allocation $\res_i \rightarrow \res_i'$ for each $i \in A$ so $q(\res_i', \candist_i) = \fair $ exactly. We are able to do this because we know that $q(\res_i, \candist_i)$ is continuous (proven in Appendix \ref{supp}) and increases monotonically from 0 to 1.

This new allocation is $\fair$-fair: $q(\res_i, \candist_i) \leq \fair $. It uses at most the same amount of resources that the optimal allocation used so it is achievable. The only case where it uses exactly the same amount of resources is when $A$ is the empty set. For $i \in A$, we know that: 
$$q(\res_i', \candist_i) = \fair = \frac{\mathbb{E}_{\candist}\br{\text{min}(\bigcan, \res_i')}}{\mathbb{E}_{\candist}[\bigcan]}$$ 
$$ \Rightarrow \quad \mathbb{E}_{\candist}\br{\text{min}(\bigcan, \res_i')} = \fair \cd \mathbb{E}_{\candist}[\bigcan] \geq \fair \cd \mathbb{E}_{\candist}\br{\text{min}(\bigcan, \res_i)}$$
This shows that the utilization of the groups in $A$ is greater than or equal to $\fair$ multiplied by the utilization of the unconstrained maximum utilization. The utilization of groups in $B$ is identical to their unconstrained utilization because their allocation has not been changed. Then:

\begin{eqnarray*}
PoF  &=& \frac{\sum_{i \in A}\mathbb{E}_{\candist}\text{min}(\bigcan, \res_i) + \sum_{i \in B}\mathbb{E}_{\candist} \text{min}(\bigcan, \res_i)}{\sum_{i \in A}\mathbb{E}_{\candist}\text{min}(\bigcan, \res_i') + \sum_{i \in B}\mathbb{E}_{\candist} \text{min}(\bigcan, \res_i)}  \\ 
& \leq &   \frac{\sum_{i \in A}\mathbb{E}_{\candist}\text{min}(\bigcan, \res_i) + \sum_{i \in B}\mathbb{E}_{\candist} \text{min}(\bigcan, \res_i)}{\fair\cd\sum_{i \in A}\mathbb{E}_{\candist}\text{min}(\bigcan, \res_i) + \sum_{i \in B}\mathbb{E}_{\candist} \text{min}(\bigcan, \res_i)} \\ 
& \leq & \frac{\sum_{i \in A}\mathbb{E}_{\candist}\text{min}(\bigcan, \res_i)}{\fair\cd\sum_{i \in A}\mathbb{E}_{\candist}\text{min}(\bigcan, \res_i)} = \frac{1}{\fair}
\end{eqnarray*}
where the inequality in the second to last step comes from the fact that $B$ could be the empty set.
\end{proof}
Note that we have required that all allocations use all of the resources, so we will need to convert the allocation calculated above into one with $\sum_i \res_i = \budget$. The supporting lemma below (proved in the Appendix) provides this final step. 
\begin{lemma}
If there exists a fractional allocation ${\bf{\res}}$ over candidate distributions $\{\candist\}$ such that $\sum_i \res_i = \bigres' < \budget$ that is $\fair$-fair, there also exists an allocation ${\bf{\res}}'$ that is $\fair$-fair and has utilization at least as large as ${\bf{\res}}$, but additionally has the property that $\sum_i \res_i = \budget$ . 
\end{lemma}

Taken together, these two lemmas show that, for every optimal allocation, we can find an allocation that is $\fair$-fair, uses all of the resources $\budget$, and has PoF less that $\frac{1}{\fair}$. This tells us that the PoF for the max-utilization $\fair$-fair allocation could be at most $\frac{1}{\fair}$ for $\fair>0$. 

\subsection[Fractional allocation unbounded]{PoF unbounded for $\fair = 0$ under fractional allocation}
However, this $\frac{1}{\fair}$ bound is undefined for $\fair=0$. The proof below shows that we can achieve arbitrarily high PoF with only $n=2$ groups, even in the case where we increase the budget $\budget$ by a constant factor in the $\fair=0$ case. 
\begin{lemma}
Suppose that in the case requiring $\fair =0$ allocation, the budget $\budget$ is multiplied by $k\geq 1$. For any $\pof > 1$, we can create a set of $n=2$ groups such that requiring $\fair=0$ fairness involves a PoF $> \pof$, assuming fractional allocation of resources.
\end{lemma}
\begin{proof}
We create 2 groups, one with distribution $\candist_1$ and one with distribution $\candist_2$, described below:
\begin{eqnarray*}
\bigcan \sim \candist_{1} && P_1(\bigcan = \can) = 
\begin{cases} 
       1 - p_1 & \can = 0 \\ 
       p_1 & \can = n_1 \\
       0 & \text{otherwise} \\
\end{cases}\\
\bigcan \sim \candist_{2} && P_2(\bigcan = \can) = 
\begin{cases} 
       1 - p_2 & \can = 0\\ 
       p_2 & \can = n_2 \\
       0 &\text{otherwise} \\
\end{cases}
\end{eqnarray*}

Set $0 < p_1 < 1$, $p_2 = \frac{p_1}{n_1}$, $n_2 = n_1^2$, and $n_1 > 2 k \cd  \pof -1$. Finally, set $\budget = n_1$. 

Because $p_1> p_2$, candidates are more likely to be present in Group 1, so the optimal allocation is $\res_1 = n_2, \res_2 = 0$. 

The $\fair=0$ fair allocation has $q(\res_1, \candist_1) = q(\res_2, \candist_2)$. In the case that $\res_i \leq n_i$, the $q$ probabilities can be calculated as follows: 
$$q(\res_i, \candist_i) = \frac{\mathbb{E}_{\candist_i}\min(\res_i, \bigcan)}{\mathbb{E}_{\candist_i}\bigcan} = \frac{\res_i\cd p_i}{n_i \cd p_i} = \frac{\res_i}{n_i}$$
As stated above, this equation only holds for $\res_i \leq n_i$. If $\res_i > n_i$, $q(\res_i, \candist_i) = 1$ always. However, we must have $\res_i \leq n_i$ in the $\fair = 0$ fair allocation: we have insufficient resources to have $q(\res_1, \candist_1) = q(\res_2, \candist_2)= 1$. 
This implies that: 
$$\frac{\res_1}{n_2} = \frac{\res_2}{n_2} \quad \res_1 + \res_2 = k \cd \budget = k \cd n_1$$
where we have incorporated the fact that the budget is $k$ times larger. 
The solution to this system of equations is:
$$\res_1 = \frac{k \cd n_1}{n_1 + n_2}\cd n_1 \quad \res_2 = \frac{k \cd n_1}{n_1 + n_2} \cd n_2 $$
This gives us a PoF as follows:
\begin{eqnarray*}
PoF &=& \frac{p_1 \cd n_1}{\frac{k \cd n_1}{n_1 + n_2}\cd n_1 \cd p_1 + \frac{k \cd n_1}{n_1 + n_2}\cd n_2 \cd p_2} = \frac{p_1}{\frac{k \cd n_1}{n_1 + n_2}\cd p_1 + \frac{k \cd n_2}{n_1 + n_2} \cd p_2}\\
&=& \frac{p_1}{\frac{k \cd n_1}{n_1 + n_1^2}\cd p_1 + \frac{k \cd n_1^2}{n_1 + n_1^2} \cd p_2}
\end{eqnarray*}
where we have used the fact that $n_2 = n_1^2$. Simplifying further gives:
$$ PoF= \frac{p_1}{\frac{k \cd n_1}{n_1+n_1^2}\cd(p_1 + n_1 \cd p_2)} = \frac{p_1}{\frac{k}{1+n_1}\cd(p_1 + n_1 \cd p_2)} = \frac{p_1 \cd (1 + n_1)}{k \p{p_1 + n_1 \cd p_2}}$$
Using the facts that $p_2 = \frac{p_1}{n_1}$ and $n_1 > 2 k \cd \pof -1 $, the bound becomes:
$$PoF = \frac{p_1 \cd (1 + n_1)}{2k \cd p_1} = \frac{1 + n_1}{2 k} > \frac{1 + 2 k \cd \pof - 1}{2 k } = \pof $$
as desired. 
\end{proof} 
Note that $p_1$ is left as a free parameter. In particular, it could be set to something decreasing in $\pof$, such as $\frac{1}{\pof}$. This demonstrates that even in the case of tail probabilities decreasing in $\pof$, it is still possible to create an example with PoF greater than $\pof$.

\section[PoF $=1$]{PoF equal to 1 for certain families of distributions}\label{pof1}
The previous section showed that with fractional allocation, the upper bound on the Price of Fairness is $\frac{1}{\fair}$. This could still be a large price to pay: if the desired level of fairness is $\fair = 0.1$, this could mean we might be obliged to reduce utilization by 90\% in order to achieve the desired level of equity. Depending on the context, this could be a very high tradeoff. However, the upper bound obtained in the previous section was derived \emph{without any reliance on the characteristics of the candidate distributions} $\{\candist\}$. In this section, we will show that for certain reasonable distributions, the Price of Fairness is much lower. In some cases, the max-utilization allocation is already $\fair=0$ fair!

\subsection[Exponential example]{Illustrative example for exponential distribution}
First, we will work out an example illustrating this phenomenon for $\ngroup = 2$ groups, both with exponential distributions. As proved in Appendix \ref{algorithms}, for this case there exists only one max-utilizing allocation, which has $F_1(\res_1) = F_2(\res_2) = \tau$, with $\res_1 + \res_2 = \budget$. This can be written as $1 - e^{-\lambda_1 \cd \res_1} = 1 - e^{-\lambda_1 \cd \p{\budget - \res_1}}$, which implies: 

$$\res_1 = \frac{ \lambda_2}{\lambda_1 + \lambda_2} \cd \budget \text{ and } \res_2 = \frac{ \lambda_1}{\lambda_1 + \lambda_2} \cd \budget$$
This allocation was calculated to be max-utilizing. We will also show that it is $\fair = 0$ fair, which means that PoF = 1. To do this, we will calculate $q(\res_i, \candist_i)$ for both groups. In general, 
\begin{eqnarray*}
q(\res, \candist) &=& \frac{\int_0^{\res}\can \cd f(\can)d\can + \res \cd (1 - F(\res))}{\mathbb{E}_{\bigcan \sim \text{exp}(\lambda)}\bigcan} = \frac{\int_0^{\res}\can \cd \lambda e^{-\lambda \can}dx + \res\cd e^{-\lambda \res}}{\frac{1}{\lambda}} \\ 
&=& \frac{\frac{1}{\lambda} \cd \p{1 - e^{-\lambda \cd \res}}}{\frac{1}{\lambda}} = 1 - e^{-\lambda \cd \res}
\end{eqnarray*}
Then, we know that $q(\res_1, \candist_1) = q(\res_2, \candist_2)$ because 
$$1 - e^{-\lambda_1 \cd \frac{\lambda_2}{\lambda_1 + \lambda_2} \cd \budget} = 1 - e^{-\lambda_2 \cd \frac{\lambda_1}{\lambda_1 + \lambda_2} \cd \budget}$$
This demonstrates that, for the given scenario, the max-utilizing allocation already satisfies $\fair=0$ fairness. However, this property is true more broadly than just in this illustrative example. In the next section, we will describe a property of probability distributions such that the max-utilizing allocation is also $\fair=0$ fair. 
\subsection{Proof of general property}
\begin{theorem}
Consider a set of continuous candidate distributions with $F_i(0) = 0$ and $f_i(\res) > 0 \ \forall \res \geq 0$. Suppose the set of candidate distributions $\{\candist_i\}$ has the following property: 
$$F_i(\res) = F_j\p{\res \cd \frac{\mathbb{E}_{\bigcan \sim \candist_j}\br{\bigcan}}{\mathbb{E}_{\bigcan \sim \candist_i}\br{\bigcan}}} \quad \forall \res \in [0, \infty) \ \forall i, j$$
Then, under fractional allocation of resources, the max-utilization allocation is already $\fair=0$ fair. In other words, $\forall i, j$
$$F_i(\res_i) = F_j(\res_j) \quad \Rightarrow \quad \frac{\mathbb{E}_{\bigcan \sim \candist_i}\min(\bigcan, \res_i)}{\mathbb{E}_{\bigcan \sim \candist_i}\br{\bigcan}} = \frac{\mathbb{E}_{\bigcan \sim \candist_j}\min(\bigcan, \res_j)}{\mathbb{E}_{\bigcan \sim \candist_j}\br{\bigcan}}$$
\end{theorem}
Intuitively, this property means that the CDFs of all of the candidate distributions $\{\candist_i\}$ are versions of the same function, scaled by the ratio of their expected values. If this property holds, it will turn out mathematically that the max-utilizing allocation is already $\fair=0$ fair.  

First, we can verify that the exponential distribution satisfies these properties. The exponential distribution has the given expectation and CDF: $$\mathbb{E}_{\bigcan \sim \candist_j}[\bigcan] = \frac{1}{\lambda_j} \quad \mathbb{E}_{\bigcan \sim \candist_i}[\bigcan] = \frac{1}{\lambda_i} \quad F_j(\res) = 1 - e^{-\lambda_j \cd \res}$$
Substitution shows that the premise of the lemmas above holds:
$$ F_j\p{\res \cd \frac{\mathbb{E}_{\bigcan \sim \candist_j}[\bigcan]}{\mathbb{E}_{\bigcan \sim \candist_i}[\bigcan]}} = 1 - e^{-\lambda_j \cd \res \cd \frac{\lambda_i}{\lambda_j}} = 1 - e^{-\lambda_i \cd \res} = F_i(\res)$$ 

Next, we will explore why this property of probability distributions leads to $PoF = 1$. As shown in Appendix \ref{algorithms}, the max-utilization allocation has the property that
\begin{itemize}
    \item $F_i(\res_i) > \tau$ only if $F_i(0) \geq \tau$. In this case, $\res_i = 0$.
    \item $F_i(\res_i) = \tau$ otherwise. 
\end{itemize}

We require $F_i(0) = 0$, so $\res_i > 0 \ \forall i$. We additionally require $f_i(\res) >0 \ \forall \res \geq 0$, which tells us that there exists only one max-utilizing allocation. Next, we examine the fairness constraint and rewrite it as:
\begin{eqnarray*}
&& \frac{\mathbb{E}_{\bigcan \sim \candist_i}\br{\bigcan \mid \bigcan \leq \res_i}\cd P_i(\bigcan \leq \res_i)}{\mathbb{E}_{\bigcan \sim \candist_i}\br{\bigcan}} + \frac{\res_i \cd P_i(\bigcan > \res_i) }{\mathbb{E}_{\bigcan \sim \candist_i}[\bigcan]} = \\
&& \frac{\mathbb{E}_{\bigcan \sim \candist_j}\br{\bigcan \mid \bigcan \leq \res_j}\cd P_j(\bigcan \leq \res_j)}{\mathbb{E}_{\bigcan \sim \candist_j}\br{\bigcan}} + \frac{\res_j \cd P_j(\bigcan > \res_j)}{\mathbb{E}_{\bigcan \sim \candist_j}[\bigcan]} \quad \forall i, j
\end{eqnarray*}
For notational convenience, we denote $m_{ij} = \frac{\mathbb{E}_{\bigcan \sim \candist_i}[\bigcan]}{\mathbb{E}_{\bigcan \sim \candist_j}[\bigcan]}$. We will prove this theorem by considering subterms of each side independently and proving that they are equal. 
\begin{lemma}
$$F_i(\res) = F_j\p{\frac{\res}{m_{ij}}} \quad \forall \res \in [0, \infty) \ \forall i, j $$
$$\Rightarrow \quad \frac{\res \cd P_i(\bigcan > \res) }{\mathbb{E}_{\bigcan \sim \candist_i}[\bigcan]} = \frac{\frac{\res}{m_{ij}}\cd P_j\p{\bigcan > \frac{\res}{m_{ij}}}}{\mathbb{E}_{\bigcan \sim \candist_j}[\bigcan]} \quad \forall \res \in [0, \infty) \ \forall i, j$$
\end{lemma}
\begin{proof}
\begin{eqnarray*}
\frac{\frac{\res}{m_{ij}}\cd P_j\p{\bigcan > \frac{\res}{m_{ij}}}}{\mathbb{E}_{\bigcan \sim \candist_j}[\bigcan]} = \frac{\frac{\res}{m_{ij}}\cd P_i(\bigcan > \res)}{\mathbb{E}_{\bigcan \sim \candist_j}[\bigcan]} = \frac{\res \cd P_i(\bigcan > \res)}{m_{ij} \cd \mathbb{E}_{\bigcan \sim \candist_j}[\bigcan]} \\ = \frac{\res\cd P_i(\bigcan > \res)}{ \mathbb{E}_{\bigcan \sim \candist_i}[\bigcan]}
\end{eqnarray*}
where the first step comes from the fact that $P_j(\bigcan > \frac{\res}{m_{ij}}) = 1 - F_j(\frac{\res}{m_{ij}}) = 1 - F_i(\res) = P_i(\bigcan > \res)$
\end{proof}
The max-utilizing allocation has the property that $F_i(\res_i) = F_j(\res_j)$. Having the additional property that $F_i(\res) = F_j\p{\frac{\res}{m_{ij}}}$ implies that $\res_j = \frac{\res_i}{m_{ij}}$. This allows us to rewrite the above lemma in terms of the max-utilizing allocations: 
$$F_i(\res_i) = F_j(\res_j) \quad \Rightarrow \quad \frac{\res_i \cd P_i(\bigcan > \res_i) }{\mathbb{E}_{\bigcan \sim \candist_i}[\bigcan]} = \frac{\res_j\cd P_j\p{\bigcan > \res_j}}{\mathbb{E}_{\bigcan \sim \candist_j}[\bigcan]}$$ 
Next, we consider the second subterm in the equality:
\begin{lemma}
$$F_i(\res) = F_j\p{\frac{\res}{m_{ij}}} \quad \forall \res \in [0, \infty) \ \forall i, j$$
$$\Rightarrow \quad \frac{\mathbb{E}_{\bigcan \sim \candist_i}[\bigcan \mid \bigcan \leq \res_i]\cd P_i(\bigcan \leq \res_i)}{\mathbb{E}_{\bigcan \sim \candist_i}\br{\bigcan}} = \frac{\mathbb{E}_{\bigcan \sim \candist_j}[\bigcan \mid \bigcan \leq \res_j]\cd P_j(\bigcan \leq \res_j)}{\mathbb{E}_{\bigcan \sim \candist_j}\br{\bigcan}}$$
\end{lemma}
\begin{proof}
What is $\mathbb{E}[\bigcan \mid \bigcan < \res]$ written out in terms of the CDF of a distribution? We will consider $\bigcan \mid \bigcan < \res$ as a different probability distribution. It has PDF and CDF as shown below:
$$f^*(t) = 
\begin{cases} 
       \frac{f(t)}{F(\res)}& 0 \leq t \leq \res \\
      0 & \text{otherwise } \\
   \end{cases} \quad F^*(t) = 
\begin{cases} 
       \frac{F(t)}{F(\res)}& 0 \leq t \leq \res \\
      1 & \text{otherwise } \\
   \end{cases}
$$
We can rewrite one of the terms above as:
\begin{eqnarray*}\mathbb{E}[\bigcan \mid \bigcan < \res] =  \int_0^{\infty}(1-F^*(t))dt = \int_0^{\res}(1-F^*(t))dt \\= \int_0^{\res}\p{1 - \frac{F(t)}{F(\res)}}dt = \res - \frac{1}{F(\res)}\int_0^{\res}F(t)dt\end{eqnarray*}
We will use this fact to evaluate the following term, substituting $\frac{\res}{m_{ij}}$ as the input $\res_j$:
\begin{eqnarray*}\mathbb{E}_{\bigcan \sim \candist_j}[\bigcan \mid \bigcan \leq \res]  \cd \frac{ F_j\p{\frac{\res}{m_{ij}}}}{\mathbb{E}_{\bigcan \sim \candist_j}[\bigcan]} = \left(\frac{\res}{m_{ij}} - \frac{\int_0^{\res/m_{ij}}F_j(t)dt}{F_j\p{\frac{\res}{m_{ij}}}}\right)\frac{F_j\p{\frac{\res}{m_{ij}}}}{\mathbb{E}_{\bigcan\sim \candist_j}[\bigcan]}
\\ 
= \left(\frac{\res\cd F_j\p{\frac{\res}{m_{ij}}} - m_{ij} \cd  \int_0^{\res/m_{ij}}F_j(t)dt}{m_{ij}\cd F_j\p{\frac{\res}{m_{ij}}}}\right)\frac{F_j\p{\frac{\res}{m_{ij}}}}{\mathbb{E}_{\bigcan \sim \candist_j}[\bigcan]}\\ 
 = \frac{\res\cd F_j\p{\frac{\res}{m_{ij}}} - m_{ij} \int_0^{\res/m_{ij}}F_j(t)dt}{m_{ij} \cd \mathbb{E}_{\bigcan \sim \candist_j}[\bigcan]} \\
 = \frac{\res \cd F_i(\res) - m_{ij} \int_0^{\res/m_{ij}}F_j(t)dt}{\mathbb{E}_{\bigcan \sim \candist_i}[\bigcan]}\end{eqnarray*}
where we have used the assumption of equality in the CDFs in the last step. Next, we note that $F_j(t) = F_i(g(t))$ where $g(t) = m_{ij}\cd t$. By the reverse chain rule, $\int_a^b f'\p{g\p{t}}\cd g'(t) dt = \br{f(g(t))}_a^b$. 
Here, $g'(t) = m_{ij}$, so if we denote the anti-derivative of $F_i$ by $H_i$, then
\begin{eqnarray*}\int_0^{\res/m_{ij}}F_j(t)\cd m_{ij} \ dt = \int_0^{\res/m_{ij}}F_i(g(t))\cd m_{ij} \ dt = \\ 
\br{H_i(g(t))}_0^{\res/m_{ij}} =\br{H_i(m_{ij}\cd t)}_0^{\res/m_{ij}} = H_i(\res) - H_i(0)  = \int_0^{\res}F_i(t)dt\end{eqnarray*}
This allows us to rewrite the item above as
$$\frac{\res \cd F_i(\res) - \int_0^{\res}F_i(t)dt}{\mathbb{E}_{\bigcan \sim \candist_i}[\bigcan]}$$
This gives a value for the righthand side of the equality we are trying to show. For the lefthand side, we use results of our previous analysis of $\mathbb{E}[\bigcan \mid \bigcan < \res]$ to find that the two terms are equivalent:
\begin{eqnarray*}
\frac{\mathbb{E}_{\bigcan \sim \candist_i}[\bigcan \mid \bigcan \leq \res] \cd F_i(\res)}{\mathbb{E}_{\bigcan \sim \candist_i}[\bigcan]} = \left(\res - \frac{1}{F_i(\res)}\int_0^{\res}F_i(t)dt\right) \cd \frac{F_i(\res)}{\mathbb{E}_{\bigcan \sim \candist_i}[\bigcan]} \\ = \frac{\res\cd F_i(\res) - \int_0^\res F_i(t)dt}{\mathbb{E}_{\bigcan \sim \candist_i}[\bigcan]}\end{eqnarray*}
\end{proof}
Taken together, these two lemmas tell us that 
$$F_i(\res) = F_j\p{\frac{\res}{m_{ij}}} \quad \forall \res \in [0, \infty)$$
$$ \quad \Rightarrow \quad \frac{\mathbb{E}_{\bigcan \sim \candist_i}\min(\bigcan, \res_i)}{\mathbb{E}_{\bigcan \sim \candist_i}[\bigcan]} = \frac{\mathbb{E}_{\bigcan \sim \candist_j}\min(\bigcan, \res_j)}{\mathbb{E}_{\bigcan \sim \candist_j}[\bigcan]} \quad \forall i, j$$
so any max-utilization allocation is already $\fair=0$ fair. 

\subsection[Distributions]{Distributions that satisfy the given premise}
We have previously shown that exponential distributions satisfy the given premise, which is perhaps not very surprising.
Exponential distributions have certain nice properties: they are memoryless, for example, and are not heavy-tailed. At first guess, it could be reasonable to guess that this phenomenon of PoF equal to 1 is only present in exponential distributions.

However, there exist further families of distributions that also satisfy the premise above and therefore have PoF equal to 1; these include Weibull distributions with the same $k$ parameter. They have expectation and CDF as given below: 
$$\mathbb{E}_{\bigcan \sim \candist_j}[\bigcan] = \lambda_j \cd \Gamma\p{1 + \frac{1}{k}} \quad \mathbb{E}_{\bigcan \sim \candist_i}[\bigcan] = \lambda_i \cd \Gamma\p{1 + \frac{1}{k}}$$
$$F_j(\res) = 1 - e^{-\p{\frac{\res}{\lambda_j}}^k} $$
Again, substitution shows that the premise of the above lemmas holds:
$$ F_j\p{\res \cd \frac{\mathbb{E}_{\bigcan \sim \candist_j}[\bigcan]}{\mathbb{E}_{\bigcan \sim \candist_i}[\bigcan]}} = 1 - e^{-\p{\frac{\res}{\lambda_j} \cd \frac{\lambda_j}{\lambda_i}}^k} =  1 - e^{-\p{\frac{\res}{\lambda_i}}^k} = F_i(\res)$$

Thus, the theorem proved in this section shows that Weibull distributions with the same $k$ parameter have PoF $=1$. This result is somewhat surprising, given that Weibull distributions in general have fewer nice properties than exponential distributions. However, Weibull distributions, along with power law distributions, which will be discussed in the next section, tend to be better models of real-life phenomena. For example, natural disasters and terrorist attacks are commonly believed to follow heavy-tailed distributions in severity \cite{Clauseteaao3580, power_law_emp}. This result shows us that there are plausibly realistic models of natural phenomena that have no trade-off between fair and max-utilizing allocations. 

\section[Power law]{PoF bounded by small constant for power law distribution}\label{powerlaw}
So far, Section \ref{generalbounds} has shown that PoF can be arbitrarily large, and Section \ref{pof1} has shown that multiple distributions have $PoF = 1$. A next logical step is to explore probability distributions with small gaps between the max-utilizing and $\fair$-fair max utilizing allocations. Specifically, we consider power law distributions. There are multiple distributions that are commonly described as \enquote{power law} distributions. The one we will use, also called the \enquote{Lomax distribution}, has PDF and CDF as below. Note that it has support for all $x\geq 0$ and a parameter $\powexp >1$ controlling its shape. 
$$f_{\candist}(\res) = \frac{\powexp}{(\res+1)^{\powexp+1}} \quad F_{\candist}(\res) =1 - \frac{1}{(1+\res)^{\powexp}} \quad \mathbb{E}_{\candist}[\bigcan] = \frac{1}{\powexp-1} $$
\begin{theorem}
Set the number of groups $\ngroup$ to some constant $\npower$. Then, the Price of Fairness from allocating $\budget$ resources across these groups is upper bounded by 
$$ \npower \cd H_{\npower}$$
where $H_{\npower}$ is the $\npower$th harmonic number.  
\end{theorem}
Note that this bound is independent of the fairness parameter $\fair$ as well as the budget $\budget$ and any distribution parameters $\powexp_i$!  

We will prove this theorem by breaking it up into several smaller lemmas. For convenience, we order the groups so that ${\powexp}_1 \geq {\powexp}_2 \ldots \geq  {\powexp}_{\npower}$. First, we write out the full PoF term: 
\begin{eqnarray*}
PoF &=& \frac{\text{Unconstrained max-utilization}}{\text{Utilization constrained to be $\fair$-fair} }\\
&=&\frac{\sum_{i=1}^{\npower} \br{\int_0^{\res_i} x \cd f_i(x)dx + \res_i \cd (1 - F_i(\res_i))}}{\sum_{i=1}^{\npower} \br{\int_0^{\res_i'} x \cd f_i(x)dx + \res_i' \cd (1 - F_i(\res_i'))}}
\end{eqnarray*}
As proved in Appendix \ref{algorithms}, because the power law distribution has positive support everywhere in the space we are considering $(\res \geq 0)$, it has a unique max-utilizing allocation and $\fair$-fair max-utilizing allocation. The PoF is largest when $\fair=0$, so that is the case we will upper-bound. 

We will denote the unconstrained max-utilization allocations by $\{\res_i\}$ and the $\fair$-fair allocations by $\{\res_i'\}$, for $i \in [\npower]$. 
\begin{lemma}\label{resmaxover}
$ \res_i \leq \frac{\budget}{\npower-i+1}\ \forall i$. 
\end{lemma}
\begin{proof}
As proved in Appendix \ref{algorithms}, the max-utilizing allocation $\{\res_i\}$ has the property that $F_i(\res_i) = F_j(\res_j)$ $ \forall \ i, j$. For the power law distributions, this property can be written as: 
$$\frac{1}{(1+\res_i)^{\powexp_i}}  = \frac{1}{(1+\res_j)^{\powexp_j}} \quad \Rightarrow \quad (1+\res_i)^{\powexp_i}  = (1+\res_j)^{\powexp_j}$$
For $\powexp_i \leq \powexp_j$, this implies $\res_i \geq \res_j$, so $\res_1 \leq \res_2 \ldots \leq \res_h$. Any max-utilizing allocation over power law distributions of this form must use the entire resource $\budget$: because the support of the distribution continues to $\infty$, any allocation using $\budget- \epsilon$ resources can be strictly improved by allocating the remaining $\epsilon$ resources. This tells us that $\res_1 + \res_2 + \ldots \res_h = \budget$. 

We can prove the upper bound by working backwards from $\npower$. $\res_{\npower} \leq \budget$ because the allocation to any one group cannot exceed the total amount of resources available. $\res_{\npower-1} \leq \res_{\npower}$ and together both of them cannot sum to more than $\budget$, so $\res_{\npower-1}\leq \frac{\budget}{2}$. Similarly, for any $\res_i$, $\res_i \leq \res_{i+1} \leq \ldots \res_{\npower}$, with $\res_i+ \res_{i+1} + \ldots \res_{\npower} \leq \budget$, so $\res_i \leq \frac{\budget}{\npower-i+1}$. 
\end{proof}
\begin{lemma}\label{resfairunder}
$\res_{\npower}' \geq \frac{\budget}{\npower}$. 
\end{lemma}
\begin{proof}
As proved in Appendix \ref{supp}, any $\fair$-fair allocation ${\bf{\res_i}}$ with $\sum_i\res_i < \budget$ can be increased to an $\fair$-fair allocation ${\bf{\res_i'}}$ with $\sum_i\res_i' = \budget$ and no decrease in utilization. This tells us that the $\fair$-fair max-utilizing allocation must use all of the resources, so $\res_1' + \res_2' + \ldots \res_{\npower}' = \budget$. 
Under the definition of $q(\res_i, \candist_i)$, 
\begin{eqnarray*} q(\res_i, \candist_i) &=& \frac{\int_0^{\res}x\cd f(x)dx + \res \cd(1- F(\res))}{\mathbb{E}_{\bigcan \sim \candist}\br{\bigcan}} \\
&=& \frac{\frac{1 - \powexp \cd \res \cd (\res+1)^{-\powexp} - (\res+1)^{-\powexp}}{\powexp-1} + \res\cd(\res+1)^{-\powexp}}{ \frac{1}{\powexp-1}}  =  1 - \frac{1}{(\res+1)^{\powexp -1}}
\end{eqnarray*}
For $\fair=0$, we require $q(\res_i', \candist_i) = q(\res_j', \candist_j)$ $\forall i, j$, which can be written as:
$$1 - \frac{1}{(\res_i'+1)^{\powexp_i -1}} = 1 - \frac{1}{(\res_j'+1)^{\powexp_j -1}} \quad \Rightarrow \quad  (\res_i'+1)^{\powexp_i -1} = (\res_j'+1)^{\powexp_j -1}$$
For $\powexp_i \geq \powexp_j$, this implies $\res_i' \leq \res_j'$. Given that we have $\npower$ terms that sum up to $\budget$ and $\res_{\npower}'$ is the largest of them, we must have $\res_{\npower}' \geq \frac{\budget}{\npower}$.
\end{proof}
While Lemma \ref{resmaxover} produces bounds for $\res_i$ $\forall i \in [\npower]$, Lemma \ref{resfairunder} only bounds the $\res_{\npower}$ case. The following lemma will demonstrate that other bounds aren't necessary, since we can rewrite all of the $\fair$-fair utilization in terms of $\res_{\npower}'$.
\begin{lemma}
The $\fair$-fair max-utilizing allocation can be rewritten entirely in terms of $\res_{\npower}'$:
\begin{eqnarray*}
\sum_{i=1}^{\npower}\br{\int_0^{\res_i'}x\cd f_i(x)dx + \res_i'\cd (1- F_i(\res_i'))} \\ 
= (\powexp_{\npower}-1) \cd \left(\sum_{i=1}^{\npower} \frac{1}{\powexp_i-1}\right) \cd \left(\int_0^{\res_{\npower}'}x\cd f_{\npower}(x)dx + \res_{\npower}'\cd (1- F_{\npower}(\res_{\npower}')) \right)
\end{eqnarray*}
\end{lemma}
\begin{proof}
$\fair=0$ fairness requires $q(\res_i', \candist_i) = q(\res_{\npower}', \candist_{\npower})$, which can be written as:
\begin{eqnarray*}
\frac{\int_0^{\res_i'}x f_i(x)dx + \res_i'(1- F_i(\res_i'))}{\frac{1}{\powexp_i-1}} = \frac{\int_0^{\res_{\npower}'}x  f_{\npower}(x)dx + \res_{\npower}'(1- F_{\npower}(\res_{\npower}'))}{\frac{1}{\powexp_{\npower}-1}} 
\end{eqnarray*}
Simple rearrangement shows that $\int_0^{\res_i'}x\cd f_i(x)dx + \res_i'\p{1- F_i(\res_i')}$ is equal to $\frac{\powexp_{\npower}-1}{\powexp_i-1}\left(\int_0^{\res_{\npower}'}x\cd f_{\npower}(x)dx + \res_{\npower}'(1- F_{\npower}(\res_{\npower}'))\right)$. 
This procedure could be repeated to obtain rewritten terms for every $\res_i'$ term with $i \ne \npower$.
\end{proof}

The above lemma allows us to rewrite the entire PoF term as:

$$\frac{\sum_{i=1}^{\npower} \br{\int_0^{\res_i} x \cd f_i(x)dx + \res_i \cd (1 - F_i(\res_i))}}{(\powexp_{\npower}-1)\cd \left(\sum_{i=1}^{\npower}\frac{1}{\powexp_i-1}\right)\left(\int_0^{\res_{\npower}'}x\cd f_{\npower}(x)dx + \res_{\npower}'\cd(1- F_{\npower}(\res_{\npower}')) \right)}
$$

Next, our approach for bounding the PoF will be to bound each term within the sum in the numerator independently. 
\begin{lemma} $PoF_{\npower} \leq \npower$ where
\begin{eqnarray*}
PoF_{\npower}:= \frac{\int_0^{\res_{\npower}}x\cd f_{\npower}(x)dx + \res_{\npower}\cd (1- F_{\npower}(\res_{\npower})) }{(\powexp_{\npower}-1)\cd \left(\sum_{i=1}^{\npower}\frac{1}{\powexp_i-1}\right)\left(\int_0^{\res_{\npower}'}x\cd f_{\npower}(x)dx + \res_{\npower}'\cd(1- F_{\npower}(\res_{\npower}')) \right)} 
\end{eqnarray*}
\end{lemma}
\begin{proof}
We can apply Lemmas \ref{resmaxover} and \ref{resfairunder} to obtain $\res_{\npower} \leq \budget$ and $\res_{\npower}' \geq \frac{\budget}{\npower}$. We additionally use the fact that
$$\int_0^{\res}x\cd f_i(x)dx + \res \cd (1 - F_i(\res)) = \frac{1}{\powexp_i-1}\p{1 - \frac{1}{(1+\res)^{\powexp_i-1}}} = \frac{1}{\powexp_i-1} q(\res, \candist_i)$$
These facts allow us to bound the equation as: 
\begin{eqnarray*}
PoF_{\npower} &\leq& \frac{\frac{1}{\powexp_{\npower}-1}\p{1 - \frac{1}{\p{1 + \budget}^{\powexp_{\npower}-1}}}}{\p{\powexp_{\npower} - 1} \cd \p{\sum_{j=1}^{\npower} \frac{1}{\powexp_j-1}} \cd \frac{1}{\powexp_h-1}\p{1 - \frac{1}{\p{1 + \frac{\budget}{\npower }}^{\powexp_h-1}}}}  \\ 
&=& \frac{1 - \frac{1}{\p{1 + \budget}^{\powexp_{\npower}-1}}}{\p{\powexp_{\npower} - 1} \cd \p{\sum_{j=1}^{\npower} \frac{1}{\powexp_j-1}} \cd \p{1 - \frac{1}{\p{1 + \frac{\budget}{\npower }}^{\powexp_h-1}}}}\\ 
&\leq&  \frac{1 - \frac{1}{\p{1 + \budget}^{\powexp_{\npower}-1}}}{1 - \frac{1}{\p{1 + \frac{\budget}{\npower }}^{\powexp_h-1}}}
\end{eqnarray*}
In the last step, to over-bound $PoF_{\npower}$, we under-bounded a term in the denominator, replacing $(\powexp_{\npower}-1)\cd \left(\sum_{i=1}^{\npower} \frac{1}{\powexp_i-1}\right)$ with $1$.

At $\budget = 0$, this fraction is undefined. We denote the numerator as $s$, the denominator as $t$, and use L'H\^{o}pital's rule to investigate the limit as $\budget \rightarrow 0$. We note that $s' = \p{\powexp_{\npower} - 1}\frac{1}{\p{1 + \budget}^{\powexp_{\npower}}}$ and  $t' = \frac{\powexp_{\npower}-1}{\npower}\frac{1}{\p{1 + \frac{\budget}{\npower}}^{\powexp_{\npower}}}$. 
As $\budget \rightarrow 0$, $s' \rightarrow \powexp_{\npower} -1$ and $t' \rightarrow \frac{\powexp_{\npower} -1}{\npower}$, so $\frac{s'}{t'} \rightarrow \npower$, which tells us that $PoF_{\npower} \rightarrow \npower$ as $\budget \rightarrow 0$. 

Next, we will show that $PoF_{\npower}$ is decreasing in $\budget$, proving that $\npower$ is an upper bound on $PoF_{\npower}$ everywhere. To do this, we will use L'H\^{o}pital's Monotone rule \cite{lhopitalrule}, which states the following:

\begin{lemma}
Given $-\infty < a < c < \infty$, and $s, t$ continuous functions differentiable on $(a,c)$ with $s(a)= t(a) = 0$ or $s(c) = t(c) = 0$ and $g(x) \ne 0 \ \forall x \in (a, c)$, then if $\frac{s'}{t'}$ is increasing (decreasing) on $(a, c)$, then $\frac{s}{t}$ is also increasing (decreasing).
\end{lemma}
The functions in $PoF_{\npower}$ satisfy this rule: $s, t$ are both continuous and differentiable, and for $a = 0$, $s(a) = t(a) = 0$. Using this rule:
$$\frac{s'}{t'} = \frac{\p{\powexp_{\npower} - 1} \frac{1}{\p{1 + \budget}^{\powexp_{\npower}}}}{\frac{\powexp_{\npower} - 1}{\npower} \frac{1}{\p{1 + \frac{\budget}{\npower}}^{\powexp_{\npower}}}} = \frac{\npower \cd \p{1 + \frac{\budget}{\npower}}^{\powexp_{\npower}} }{\p{1 + \budget}^{\powexp_{\npower}}} = \npower \cd \p{\frac{1 + \frac{\budget}{\npower}}{1 + \budget}}^{\powexp_{\npower}}$$
This term is decreasing in $\budget$ because the denominator $1 +\budget$ is increasing faster than the numerator $1 + \frac{\budget}{\npower}$ is increasing. This tells us that $\frac{s'}{t'}$ is decreasing in $\budget$, so $PoF_{\npower}$ is decreasing as well. The upper bound of $\budget = 0$ is an upper bound for all of $PoF_{\npower}$. 
\end{proof}

Next, we consider a bound for any $i \ne \npower$. 
\begin{lemma} For $i \in \br{\npower -1}$, $PoF_{i} \leq \frac{\npower}{\npower - i + 1} $, where
\begin{eqnarray*}
PoF_{i} := \frac{\int_0^{\res_i}x\cd f_i(x)dx + \res_i(1- F_i(\res_i)) }{(\powexp_{\npower}-1)\cd\left(\sum_{j=1}^{\npower}\frac{1}{\powexp_j-1}\right)\left(\int_0^{\res_{\npower}'}x\cd f_{\npower}(x)dx + \res_{\npower}'(1- F_{\npower}(\res_{\npower}')) \right)}
\end{eqnarray*}
\end{lemma}
\begin{proof}Using bounds $\res_i \leq \frac{\budget}{\npower-i+1}$ and $\res_{\npower}' \geq \frac{\budget}{\npower}$ from Lemmas \ref{resmaxover} and \ref{resfairunder}, we substitute and rewrite: 
\begin{eqnarray*}
PoF_i &\leq & \frac{\frac{1}{\powexp_i-1}\p{1 - \frac{1}{\p{1 + \frac{\budget}{\npower - i + 1 }}^{\powexp_i-1}}}}{\p{\powexp_{\npower} - 1} \cd \p{\sum_{j=1}^{\npower} \frac{1}{\powexp_j-1}} \cd \frac{1}{\powexp_h-1}\p{1 - \frac{1}{\p{1 + \frac{\budget}{\npower }}^{\powexp_h-1}}}}\\ 
&=& \frac{1 - \frac{1}{\p{1 + \frac{\budget}{\npower - i + 1 }}^{\powexp_i-1}}}{\p{\powexp_{i} - 1} \cd \p{\sum_{j=1}^{\npower} \frac{1}{\powexp_j-1}} \cd \p{1 - \frac{1}{\p{1 + \frac{\budget}{\npower }}^{\powexp_h-1}}}} \\ 
&<& \frac{1 - \frac{1}{\p{1 + \frac{\budget}{\npower - i + 1 }}^{\powexp_i-1}}}{\p{\frac{\powexp_{i} - 1}{\powexp_{\npower}-1} + 1} \cd \p{1 - \frac{1}{\p{1 + \frac{\budget}{\npower }}^{\powexp_h-1}}}}
\end{eqnarray*}
At $\budget = 0$, this is again undefined. Using the same process of rewriting the fraction as $\frac{s}{t}$, we get 
$$s' = \frac{\powexp_i-1}{\npower - i + 1} \frac{1}{\p{1 + \frac{\budget}{\npower - i +1 }}^{\powexp_i}}\quad t' = \frac{\powexp_{\npower} - 1 + \powexp_i-1}{\npower} \frac{1}{\p{1 + \frac{\budget}{\npower}}^{\powexp_{\npower}}} $$
Using L'H\^{o}pital's rule again, as $\budget \rightarrow 0$, $\frac{s'}{t'}$ goes to a ratio of: 
$$\frac{\frac{\powexp_i - 1}{\npower - i + 1}}{\frac{\powexp_{\npower} - 1 +\powexp_i - 1}{\npower}} = \frac{\powexp_i - 1}{\powexp_{\npower} - 1 + \powexp_i - 1} \cd \frac{\npower}{\npower - i + 1} = \frac{1}{\frac{\powexp_{\npower} - 1}{\powexp_i-1} + 1} \cd \frac{\npower}{\npower - i + 1}$$
$\frac{\powexp_{\npower} - 1}{\powexp_i-1}$ is decreasing in $\powexp_i$, so this overall term is increasing in $\powexp_i$. $\powexp_i \rightarrow \infty$ upper-bounds this fraction by $\frac{1}{0 + 1} \cd \frac{\npower}{\npower - i + 1} = \frac{\npower}{\npower - i + 1}$
which provides the limit as $\budget \rightarrow 0$. 
Again, we use L'H\^{o}pital's Monotone rule to show that $\frac{s}{t}$ is decreasing in $\budget$ and that the $\budget = 0$ upper bound is sufficient  $\forall \budget$. This tells us that $\frac{s'}{t'}$ equals
$$\frac{\frac{\powexp_i -1}{\npower - i + 1} \frac{1}{\p{1 + \frac{\budget}{\npower - i + 1}}^{\powexp_i}}}{\frac{\powexp_{\npower} - 1 + \powexp_i-1}{\npower} \frac{1}{\p{1 + \frac{\budget}{\npower}}^{\powexp_{\npower}}}} =\frac{\powexp_i -1}{\npower - i + 1} \cd \frac{\npower}{\powexp_{\npower} - 1 + \powexp_i-1} \cd \frac{\p{1 + \frac{\budget}{\npower}}^{\powexp_{\npower}}}{\p{1 + \frac{\budget}{\npower - i + 1}}^{\powexp_i}}$$
The term in the denominator has a base that grows more quickly in $\budget$ (because $1 + \frac{\budget}{\npower - i + 1}$ increases in $\budget$ faster than  $1 + \frac{\budget}{\npower}$) and a larger exponent (because $\powexp_i \geq \powexp_{\npower}$). Because of this, the denominator grows more quickly in $\budget$, so $\frac{s'}{t'}$ is decreasing in $\budget$. L'H\^{o}pital's Monotone rule tells us that $\frac{s}{t}$ is decreasing in $\budget$ as well, and so the over bound at $\budget = 0$ is an over bound for the entire function. 
\end{proof}

These lemmas, taken together, give us a bound for the overall PoF in terms of $H_{\npower}$, the $\npower$th harmonic number:
$$\npower + \sum_{i=1}^{\npower-1}\left[\frac{\npower}{\npower-i+1}  \right] = \npower \cd H_{\npower}$$

We also investigated this bound empirically by directly calculating the PoF for groups with different power law distributions. We found that the PoF is typically much lower than this theoretical bound, frequently around 1.1. This result tells us that, in general, the most max-utilizing allocations are already very close in utilization to $\fair=0$ fair solutions. Overall, this provides an optimistic outlook for resource allocation in the numerous real-world scenarios modeled by power law distributions.

\section{Conclusion}

In this work we have analyze a problem in fair resource allocation from the recent literature \cite{Elzayn_2019}.  Starting from a set of open questions around the relationship between fair allocations and maximum-utilization allocations, we obtained upper bounds on the gap between the utilization of optimal and fair solutions --- formalized as the {\em Price of Fairness} (PoF).  In addition to bounds for general distributions, we showed that a number of natural families of distributions exhibit no gap at all: there exist optimal allocations that also achieve perfect fairness guarantees. Finally, we established constant upper bounds for power-law distributions.

There are a number of further directions suggested by this work.  To begin with, it would be interesting to try characterizing the set of distributions for which the Price of Fairness is equal to 1; this is a fundamental distributional property that applies more broadly than it initially appears.  It would also be valuable to extend the set of distributions for which we have the techniques to prove constant upper bounds on the Price of Fairness. 

Additionally, future work could explore different definitions of fairness. If we consider models where candidates are not interchangeable and persist from one time state to another, it might be considered more \enquote{fair} to allocate resources preferentially to those who have been waiting longest for resources. Finally, the model so far largely assumes a fixed amount of resources. In some cases, though, the amount of resources could be increased or decreased and it could be enlightening to consider how modifying the availability of resources might impact the tradeoff between utilization and fairness. By developing a richer understanding of the behavior of different classes of distributions in this model, we can further our understanding of the interplay between fairness and utilization when we allocate resources with this type of uncertainty.

\begin{acks}
This work was in part supported by a Simons Investigator Award,
a grant from the MacArthur Foundation,
and NSF grants DGE-1650441, CCF-1740822, and SES-1741441. We are grateful to Manish Raghavan, the \href{https://aipp.cis.cornell.edu/}{AI in Policy and Practice} working group at Cornell, and Raunak Kumar for invaluable discussions. We would also like to thank four anonymous reviewers for their kind and insightful feedback. 
\end{acks}

\bibliographystyle{ACM-Reference-Format}
\bibliography{sample-base}

\clearpage
\appendix

\section{Summary table}\label{table}

Table \ref{tab:overview} displays a high-level summary of the results of this paper. The main point is to highlight the distinction between discrete vs. continuous \emph{resource allocation} and discrete vs. continuous \emph{candidate probability distributions}. The former distinction produces more substantive differences in results, whereas the latter distinction mainly requires different tools of analysis but similar results. 

\begin{table*}[]
    \centering
    \begin{tabular}{p{2cm}|p{6cm}|p{6cm}}
         & Discrete candidate probability  distribution & Continuous candidate probability distribution\\ \hline
         Discrete \newline resource \newline  allocation &          \begin{itemize}
             \item This situation is the main focus of \citein{Elzayn_2019}..
             \item \citein{ganchev2009censored} provides a greedy allocation algorithm for this situation.
             \item Appendix \ref{supp} uses an example of this type to show unbounded PoF for discrete allocation in general
         \end{itemize} &  
          \\  \hline
         Fractional resource allocation &
         \begin{itemize}
             \item Appendix \ref{algorithms} shows that the same algorithm from \citein{ganchev2009censored} returns the max-utilizing allocation in this situation.  
             \item Section \ref{generalbounds} uses an example of this type to show that the PoF can be arbitrarily large for $\fair =0$ in fractional allocation case.
         \end{itemize} & 
         \begin{itemize}
             \item Appendix \ref{algorithms} provides an algorithm to calculate a max-utilizing allocation in this situation. 
             \item Section \ref{pof1} shows that exponential \& Weibull distributions have PoF $=1$ in this situation.  
             \item Section \ref{powerlaw} shows that power law distributions have PoF equal to a small constant in this situation. 
         \end{itemize}    \\ 
         & \multicolumn{2}{p{10cm}}{\emph{Results for both discrete \& continuous distributions}: \begin{itemize}
            \item Section \ref{generalbounds} shows that $\forall \ \fair > 0$, it is possible to create an $\fair$-fair allocation with PoF $\leq 1/\fair$.
         \end{itemize}} 
    \end{tabular}
    \caption{Summary of results}
    \label{tab:overview}
    \vspace{-6mm}
\end{table*}
\section{Supplementary lemmas}\label{supp}
\begin{lemma}
Under integer allocation of resources, for any desired level of fairness $\fair < 1$ and any desired PoF $\pof$, we can create a set of $\ngroup$ groups, candidate distributions across these groups $\{\candist_i\}$, and a level of resources $\budget$ such that the PoF is greater than $\pof$.
\end{lemma}

\begin{proof}
We will use the inputs $\fair$ and $\pof$ to design a set of candidate distributions such that the PoF will be $ > \pof$. The building block of this proof will be two distinct kinds of candidate distributions. Both are discrete and deterministic. 
$$\bigcan_i \sim \candist_{high} \quad P(\bigcan_i = \can) = 
\begin{cases} 
       1 & \can = n  \\
       0 & \text{ otherwise} \\
\end{cases}
$$
$$\bigcan_i \sim \candist_{low} \quad P(\bigcan_i = \can) = 
\begin{cases} 
       1 & \can = n' \\
       0 & \text{ otherwise} \\
\end{cases}
$$
We will set the parameters based on the input values: set $n = \ceil{ \frac{1}{\fair}}$ and $n' =\floor{\frac{1}{\fair}} $ if $\frac{1}{\fair}$ is not an integer and $\frac{1}{\fair} -1$ if it is. Note that this implies that $n \geq \frac{1}{\fair}$ and $n' < \frac{1}{\fair}$. We will find it useful to define a parameter $m'$ to be any integer such that
$$m' > \pof \cd \frac{\fair (1+\fair)}{1-\fair}$$
As a reminder, $0 \leq \fair < 1$, $\pof \geq 1$ so this number is always positive. We will set the total number of resources to be equal to 
$$\budget = n +n' \cd m' $$
Finally, we will set $\{\candist_i\}$ so that we have 1 of the $\candist_{high}$ groups and $m = \budget + 1$ of the $\candist_{low}$ groups, for $\ngroup = \budget + 2$ overall.

Next, we calculate the PoF. Under the max-utilization allocation, we have sufficient resources to have $\budget$ candidates receive the resource. One allocation that achieves this upper bound is to put $n$ resources in the $\candist_{high}$ group and $n'$ resources in each of $m'$ of the $\candist_{low}$ groups. 

Next, we consider $\fair$-fair allocation. $\ngroup = \budget + 2 > \budget$, so we have insufficient resources to put 1 resource in each group. This means that the probability of a candidate receiving the resource will be 0 for some group, so the probability of a candidate receiving the resource in any group must be $\leq \fair$. 

By construction, putting 1 resource in any of the $\candist_{low}$ groups means that a candidate has a $\frac{1}{n'}> \fair$ chance of receiving the resource, so a max-utilizing fair allocation must have 0 resources at each of the $\candist_{low}$ groups. The highest allocation to $\candist_{high}$ that has $q(\res_i, \candist_{high}) \leq \fair$ is $\lfloor \fair n \rfloor$. 

The PoF is: 
$$\frac{n + m'\cd n'}{\lfloor \fair n \rfloor }$$
We will place a lower bound of $\pof$ on the overall fraction by putting a lower bound on the numerator and an upper bound on the denominator. For the lower bound on the numerator, we have: 
$$n + m' \cd n' \geq \frac{1}{\fair} + m'\cd n' \geq \frac{1}{\fair} + m'\cd \p{\frac{1}{\fair} - 1} > \frac{1}{\fair} + \pof \cd \frac{\fair \cd  (1+\fair)}{1-\fair} \cd  \p{\frac{1}{\fair} - 1}$$
We can simplify the overall terms and strategically rewrite:
$$ \frac{1}{\fair} + \pof \cd \frac{\fair  \cd (1+\fair)}{1-\fair} \cd \frac{1 - \fair}{\fair} = \frac{1}{\fair} + \pof \cd (1 + \fair) = \frac{1}{\fair}\left(1 + \pof \cd \fair  \cd  (1 + \fair)\right)$$
For an upper bound on the denominator, we know that: 
$$\floor{\fair n} \leq \fair \cd  n  = \fair \ceil{ \frac{1}{\fair}} < \fair \cd \p{\frac{1}{\fair} + 1} = \fair + 1$$
where the last inequality comes from the fact that $\lceil \frac{1}{\fair}\rceil < \frac{1}{\fair}+1$. 
Recombining the fractions, we get: 
$$\text{PoF} > \frac{\frac{1}{\fair}\left(1 + \pof \cd \fair \cd (1 + \fair)\right)}{\fair + 1} > \frac{\frac{1}{\fair}\br{\fair \cd \pof \cd (1+\fair)}}{\fair + 1} = \pof$$
So the PoF is strictly greater than $\pof$, as desired. This result shows that for any PoF level $\pof$, we can construct a set of candidate distributions with a higher PoF, so the PoF is unbounded. 
\end{proof}

\begin{lemma}
For the case where the allocation is equal to $ \res + \epsilon$, with $\res \in \mathbb{Z}$ and $0<\epsilon <1$ over a discrete distribution $\candist$, the following perspectives provide the same expected utilization: 
\begin{itemize}
    \item Viewing the allocation as deterministically providing $\res + \epsilon$ resources. 
    \item Viewing it as providing $\res$ resources with probability $1- \epsilon$ and $\res+1$ resources with probability $\epsilon$. 
\end{itemize} 
\end{lemma}
\begin{proof}
If we view this allocation deterministically, we could convert $\candist$ to a continuous distribution that has point mass probability at the integers and 0 mass elsewhere. In this case, the expected utilization would be equal to: 
$$\sum_{\can=0}^{\res}\can\cd f(\can) + (1-F(\res))\cd(\res + \epsilon)$$
We would like our interpretation under the probablistic interpretation to simplify down to an equivalent term. This interpretation looks like: 
\begin{eqnarray*}
\mathbb{E}_{\resdist}\mathbb{E}_{\candist}\min(\bigcan,\bigres) = (1-\epsilon)\cd \left(\sum_{\can=0}^{\res}\can\cd f(\can) + \res\cd \left(\sum_{\can=\res+1}^{\infty}\can\cd f(\can)\right)\right) + \\  \epsilon \cd\left(\sum_{\can=0}^{\res+1}\can\cd f(\can) + (\res+1)\left(\sum_{\can=\res+2}^{\infty}\can\cd f(\can)\right)\right)
\end{eqnarray*}
Combining the first terms in the sum first: 
\begin{eqnarray*}
&&\left(\sum_{\can=0}^{\res}\can\cd f(\can)\right) + \epsilon \cd (\res+1)\cd f(\res+1) \\ 
&& +\ (1-\epsilon)\cd \res\cd \left(\sum_{\can=\res+1}^{\infty}\can\cd f(\can)\right)+ \epsilon\cd (\res+1)\cd\left(\sum_{\can=\res+2}^{\infty}\can\cd f(\can)\right)
\end{eqnarray*}
Absorbing the leftover $\epsilon\cd(\res+1)\cd f(\res+1)$ term inside: 
\begin{eqnarray*}
& = & \left(\sum_{\can=0}^{\res}\can\cd f(\can)\right) + (1-\epsilon)\cd \res\cd \left(\sum_{\can=\res+1}^{\infty}\can\cd f(\can)\right) \\ 
& & + \  \epsilon\cd (\res+1)\cd\left(\sum_{\can=\res+1}^{\infty}\can\cd f(\can)\right) \\ 
& = & \left(\sum_{\can=0}^{\res}\can\cd f(\can)\right) + \left(\sum_{\can=\res+1}^{\infty}\can\cd f(\can)\right)\p{(1-\epsilon)\cd \res + \epsilon\cd (\res+1)}\\
 &= &\left(\sum_{\can=0}^{\res}\can\cd f(\can)\right) + \left(1- F(\res)\right)(\res +\epsilon)
\end{eqnarray*}
As desired. 
\end{proof}

\begin{lemma}
The probability of service function $q({\bf{\res}}, \{\candist\})$ is continuous. 
\end{lemma}

\begin{proof}
We work directly from the formal definition of continuity, which is that: 
$$\forall \epsilon >0 \ \exists \delta > 0 \ \text{ s.t. } \vert x - y \vert < \delta \ \Rightarrow \ \vert f(x) - f(y) \vert < \epsilon$$ 
The probability of service function is $q(\res_i, \candist_i) = \frac{\mathbb{E}_{\bigcan \sim \candist} \min(\bigcan, \res)}{\mathbb{E}_{\bigcan \sim \candist}\br{\bigcan}}$. 

We first consider the case of discrete $\candist$. If we are trying to prove continuity at a particular point $\res$, we pick $\delta < \epsilon \cd \mathbb{E}_{\bigcan \sim \candist}\br{\bigcan}$. We additionally require $\lfloor \res \rfloor = \lfloor \res + \delta \rfloor$. This is always possible if we make $\delta$ small enough: if $\res$ is an integer, then any $\delta < 1$ works. If $\res$ is not an integer, then there is some positive distance $d$ between it and the next integer. Then, any $\delta < d$ will work. Then 
\begin{eqnarray*}
&& q(\res + \delta, \candist) - q(\res, \candist) \\
&=& \frac{1}{\mathbb{E}_{\bigcan \sim \candist}\br{\bigcan}} \cd  \left(\sum_{\can=0}^{\lfloor \res + \delta \rfloor}\can\cd P(\bigcan = \can) +(\res + \delta)P(\can > \lfloor \res + \delta \rfloor)\right) \\  
& &   - \frac{1}{\mathbb{E}_{\bigcan \sim \candist}\br{\bigcan}} \cd  \left(\sum_{\can=0}^{\lfloor \res \rfloor}\can \cd P(\bigcan = \can) + \res \cd P(\bigcan > \lfloor \res \rfloor)\right) \\ 
&=& \frac{1}{\mathbb{E}_{\bigcan \sim \candist}\br{\bigcan}} \cd (\res + \delta - \res)P(\bigcan > \lfloor \res \rfloor) \\ 
&<&  \frac{1}{\mathbb{E}_{\bigcan \sim \candist}\br{\bigcan}} \cd \epsilon \cd \mathbb{E}_{\bigcan \sim \candist}\br{\bigcan} \cd  P(\bigcan > \lfloor \res \rfloor) \\ 
&\leq& \epsilon
\end{eqnarray*}
Next, we consider the continuous $\candist$ case. Again, we pick $\delta < \epsilon \cd \mathbb{E}_{\bigcan \sim \candist}\br{\bigcan}$. 
\begin{eqnarray*}
&& q(\res + \delta, \candist) - q(\res, \candist) \\  
&=&\frac{1}{\mathbb{E}_{\bigcan \sim \candist}\br{\bigcan}} \cd \left(\int_0^{\res+\delta}\can \cd f(\can)d\can + (\res + \delta) \int_{\res+\delta}^{\infty} f(\can) d\can \right)\\ 
 && -  \frac{1}{\mathbb{E}_{\bigcan \sim \candist}\br{\bigcan}} \cd \left(\int_0^{\res}\can\cd f(\can)d\can + \res \int_{\res}^{\infty} f(\can) d\can\right)\\ 
&=& \frac{1}{\mathbb{E}_{\bigcan \sim \candist}\br{\bigcan}} \cd \p{\int_{\res}^{\res+\delta}\can\cd f(\can) d\can + (\res + \delta)\int_{\res+\delta}^{\infty}f(\can)d\can} \\
&& \frac{1}{\mathbb{E}_{\bigcan \sim \candist}\br{\bigcan}} \cd \p{- \res \int_{\res}^{\infty}f(\can)d\can}
\end{eqnarray*}
We add and subtract $\delta \cd  \int_{\res}^{\res+\delta}f(\can)d\can$
and split up one of the terms in the middle:
\begin{eqnarray*}
& =& \frac{1}{\mathbb{E}_{\bigcan \sim \candist}\br{\bigcan}}  \p{\int_{\res}^{\res+\delta}\can f(\can) d\can + (\res + \delta)\int_{\res+\delta}^{\infty}f(\can)d\can} 
\\
&& + \ \frac{1}{\mathbb{E}_{\bigcan \sim \candist}\br{\bigcan}} \cd \p{- \res \int_{\res}^{\res + \delta}f(\can)d\can  - \res \int_{\res + \delta}^{\infty}f(\can)d\can}
\\
&&+ \ \frac{1}{\mathbb{E}_{\bigcan \sim \candist}\br{\bigcan}} \cd \p{ \delta  \int_{\res}^{\res+\delta}f(\can)d\can - \delta  \int_{\res}^{\res+\delta}f(\can)d\can }
 \end{eqnarray*}
\begin{eqnarray*}
&=& \frac{1}{\mathbb{E}_{\bigcan \sim \candist}\br{\bigcan}} \cd \p{\int_{\res}^{\res+\delta}\can f(\can) d\can + (\res + \delta)\int_{\res+\delta}^{\infty}f(\can)}\\
&& + \  \frac{1}{\mathbb{E}_{\bigcan \sim \candist}\br{\bigcan}} \cd \p{- (\res + \delta) \int_{\res}^{\res + \delta}f(\can)d\can}\\
&& + \  \frac{1}{\mathbb{E}_{\bigcan \sim \candist}\br{\bigcan}} \cd \p{- \res \int_{\res + \delta}^{\infty}f(\can)d\can + \delta  \int_{\res}^{\res+\delta}f(\can)d\can}
\end{eqnarray*}
Next, we note that $\int_{\res}^{\res+\delta}\can\cd f(\can)d\can \leq (\res+\delta)\int_{\res}^{\res+\delta}f(\can)d\can$, so
$(\res+\delta)\int_{\res}^{\res+\delta}f(\can)d\can -\int_{\res}^{\res+\delta}\can\cd f(\can)d\can \leq 0$. We can drop this difference to get an upper bound
\begin{eqnarray*}
 &\leq&  \frac{1}{\mathbb{E}_{\bigcan \sim \candist}\br{\bigcan}} \cd  (\res + \delta)\int_{\res+\delta}^{\infty}f(\can)d\can \\
 && +  \ \frac{1}{\mathbb{E}_{\bigcan \sim \candist}\br{\bigcan}} \p{- \res \int_{\res + \delta}^{\infty}f(\can)d\can + \delta \cd  \int_{\res}^{\res+\delta}f(\can)d\can} \\ 
 &=&  \frac{1}{\mathbb{E}_{\bigcan \sim \candist}\br{\bigcan}} \cd \p{ \delta\int_{\res+\delta}^{\infty}f(\can)d\can + \delta \cd  \int_{\res}^{\res+\delta}f(\can)d\can}\\
 &=& \frac{1}{\mathbb{E}_{\bigcan \sim \candist}\br{\bigcan}} \cd \delta \cd  \int_{\res}^{\infty}f(\can)d\can\\
 &\leq&  \frac{1}{\mathbb{E}_{\bigcan \sim \candist}\br{\bigcan}} \cd \epsilon \cd \mathbb{E}_{\bigcan \sim \candist}\br{\bigcan} \cd  \int_{\res}^{\infty}f(\can)d\can\\
&=& \epsilon \int_{\res}^{\infty}f(\can)d\can < \epsilon
\end{eqnarray*}
\end{proof}

\begin{lemma}
If there exists a fractional allocation ${\bf{\res}}$ over candidate distributions $\{\candist\}$ such that $\sum_i \res_i = \bigres' < \budget$ that is $\fair$-fair, there also exists an allocation ${\bf{\res}}'$ that is $\fair$-fair and has utilization at least as large as ${\bf{\res}}$, but additionally has the property that $\sum_i \res_i = \budget$ . 
\label{increase_res}
\end{lemma}

\begin{proof}
This will be a proof by construction: we will construct an allocation that satisfies the desired properties. At each step, if accomplishing the step would require more than all of the remaining resources, then accomplish the step \enquote{partially} (as much as could be done with the remaining resources), and then proceed to the conclusion.
\begin{enumerate}
\item First, pick $\beta^* = \max_i q(\res_i, \candist_i)$. For any groups $j$ with $\beta_j < \beta^{*}$, increase their allocation $\res_j$ until $\beta_j = \beta^{*}$. This maintains $\fair$-fairness because it only decreases distance between the probability of service between groups. Additionally, it either increases or holds constant the utilization. We know that this is achievable (assuming sufficient resources) due to the continuity of $q$. If $\beta^* = 1$, skip to the last step. 
\item Next, arbitrarily pick some $\fair_1< \fair$ such that $\beta^{*} + \fair_1 \leq 1$. Pick some $i \in \br{\ngroup}$. Increase $\res_i$ until $q(\res_i', \candist_i) = \beta^{*} + \fair_1$. Again, this is $\fair$-fair, utilization-increasing, and achievable assuming sufficient resources. 
\item Increase the allocation to all of the other groups until they reach $\beta^{*} + \fair_1$. Repeat this process $M$ times until $\beta^{*} + \sum_{i = 1}^{M}\fair_i$ = 1.  This preserves $\fair$-fairness, is utilization-increasing, and is achievable assuming sufficient resources. 
\item If this step is reached, the previous steps have resulted in an allocation of $\bigres''<\budget$, but with $q(\res_i, \candist_i) = 1 \ \forall i$. In this case, pick $i \in \br{\ngroup}$ arbitrarily and put the remaining $\budget - \bigres''$ resources there. This does not change its $q(\res_i', \candist_i)$ value because all of the candidates already receive resources. It also does not change utilization, but it achieves $\sum_i \res_i = \budget$. 
\end{enumerate}
After this algorithm has been run, it results in an allocation that uses all of the resources, is $\fair$-fair, and has utilization equal to or greater than the existing allocation. 
\end{proof}

\section[Fractional allocation]{Calculating max-utilizing fractional allocation}\label{algorithms}
One contribution of this work is to consider the fractional allocation of resources. In this section, we compute max-utilizing fractional resource allocations. The first section will focus on continuous candidate distributions, while the second section will analyze fractional allocation across discrete probability distributions. 

\subsection[Fractional \& continuous]{Fractional allocation over continuous candidate distributions}
The goal of this section will be to prove the following theorem: 
\begin{theorem}
A max-utilization fractional allocation ${\bf{\res}}$ with resources $\budget$ across groups with continuous distributions $\{\candist_i\}$ is one that satisfies: 
\begin{itemize}
    \item All of the CDFs $F_i(\res_i)$ across the groups are equal to the same value, $\tau$, except in the case where $F_i(0) > \tau$: in these cases, $\res_i = 0$.
    \item If the PDF of $\candist$ has nonzero support for all $\res>0$, there will only be one max-utilization solution.
\end{itemize}
\end{theorem}
Given a set of resources $\budget$ and probability distributions over candidates $\{\candist_i\}$, it is straightforward to calculate the max-utilization allocation by running a binary search between CDF values of 0 and 1 across the groups and finding the value where $\sum_i \res_i = \budget$. 

In order to prove this theorem, first we will need to prove some properties of the functions we are working with. Our goal will be to formulate this as a convex optimization problem. 
\begin{lemma}
The derivative of the utilization function $U({\bf{\res}}, \{\candist\})$ with respect to $\res_i$ is $P_i(\bigcan>\res_i) = 1 - F_i(\res_i)$ for continuous functions. 
\end{lemma}
\begin{proof}
If we have $\ngroup$ groups, the overall utilization is equal to 
$$U(\res_1, \ldots \res_{\ngroup}) = \sum_{i=1}^{\ngroup}\mathbb{E}_{\bigcan\sim \candist_i}\text{min}(\bigcan, \res_i)$$
The derivative with respect to $\res_i$ is 0 for all but one of the terms in this sum: 
\begin{eqnarray*}
\frac{d U(\res_1, \ldots \res_{\ngroup})}{d\res_i} &=& \frac{d}{d\res_i}\mathbb{E}_{\bigcan \sim \candist_i}\text{min}(\bigcan, \res_i) \\ &=& \frac{d}{d\res_i}\left[\int_0^{\res_i}\can\cd f_i(\can)d\can + \res_i \cd (1 - F_i(\res_i))\right]
\end{eqnarray*}
We evaluate the derivative starting with the first term. By the fundamental theorem of calculus, 
$$\frac{d}{d\res_i}\int_0^{\res_i}\can\cd f_i(\can)d\can = \res_i\cd f_i(\res_i)$$
For the second term:
\begin{eqnarray*}
\frac{d}{d\res_i}\res_i\p{1 - F_i(\res_i)} &=& \frac{d}{d\res_i}(\res_i - \res_i\cd F_i(\res_i)) \\  
&=& 1 - (F_i(\res_i) + \res_i\cd f_i(\res_i)) \\
&=& 1 - F_i(\res_i) - \res_i \cd f_i(\res_i)
\end{eqnarray*}
where we have used the fact that the CDF's derivative is the PDF. The sum of both terms together is $1-F_i(\res_i)$. 
\end{proof}

\begin{lemma}
The utilization function $U({\bf{\res}})$ is concave. If the PDF is strictly greater than 0, then the utilization function is strictly concave. 
\end{lemma}
\begin{proof}
Given a function $f$ with Hessian $H$, $f$ is concave if and only if ${\bf{a}}^TH{\bf{a}} \leq 0$ for all ${\bf{a}} \in \mathbb{R}^n/\{{\bf{0}}\}$. The lemma above showed that the first derivative of $U({\bf{\res}})$ with respect to $\res_i$ is $[0, \ldots 1 - F(\res_i), \ldots 0]$. Given this result, the Hessian is a matrix equal to $-f(\res_i)$ in the $i$th row and $i$th column pair, with 0 elsewhere. This tells us that 
$${\bf{a}}^TH{\bf{a}} = \sum_{i=1}^n-a_i^2f_i(\res_i)$$ 
$f_i(\res_i) \geq 0$, so this overall sum is $\leq 0$. The inequality is strict when $f_i(\res_i) = 0$ for some $i, \res_i$ combination: then, a ${\bf{a}}$ vector that has 0s everywhere but in the $i$th entry will have ${\bf{a}}^TH{\bf{a}} = 0$. If $f_i(\res_i) > 0$ always, then the utility function is strictly concave, otherwise, it is concave. 
\end{proof}

Finally, we will put all of these pieces together to show how to find a max-utilization allocation for continuous probability distributions. We will do this by formulating the problem as a convex optimization problem satisfying the Karush--Kuhn--Tucker (KKT) conditions.

In the optimization problem, our goal is:
$$\text{minimize } - \sum_{i = 1}^{\ngroup}\mathbb{E}_{\bigcan_i \sim \candist_i}\min(\bigcan_i, \res_i)$$
$$\text{subject to } -\res_i \leq 0 \ \forall i \quad \left(\sum_{i=1}^{\ngroup} \res_i \right) - \budget = 0$$
Following the example of Section 5.5 in \citein{boyd2004convex}, we require that the inequality equations are convex and the equality constraint is affine. A function $f$ is convex if
$$f(t\cd x_1 + (1-t) \cd x_2) \leq t \cd f(x_1) + (1-t)\cd f(x_2)$$
For the terms we are considering, the inequality function is $f(x)= -x$, so plugging in for this on both sides yields
 $$-t \cd x_1 - (1-t) \cd x_2 = -t \cd x_1 -(1-t) \cd x_2$$
satisfying convexity. A function is affine if it can be written as a linear transformation and a translation: the equality $\left(\sum_{i=1}^{\ngroup}r_i \right) - \budget = 0$ already satisfies those requirements. 

Then, the KKT results tell us that any points with resource allocation $\tilde \res_i$, inequality constraints $\tilde \lambda_i$ and equality constraint $\tilde \tau$ that satisfy the below equations are primal and dual optimal, with 0 duality gap.
\begin{eqnarray*}
-\tilde \res_i &\leq& 0 \quad \forall i \in [\ngroup] \\ \left(\sum_{i=1}^{\ngroup} \tilde \res_i\right) - \budget &=& 0 \\ 
\tilde \lambda_i &\geq& 0 \quad \forall i \in [\ngroup]  \\ 
-\tilde \lambda_i \tilde \res_i &=& 0 \quad \forall i \in [\ngroup] 
\end{eqnarray*}
$$\nabla_{\res} \left(\sum_{i=1}^{\ngroup}\mathbb{E}_{\bigcan_i \sim \candist_i}\min(\bigcan_i, \res_i)\right) + \sum_{i=1}^{\ngroup}\nabla_{\res} ( \lambda_i(- \res_i)) + \nabla_{\res} \tau \cd \left( \p{\sum_{i=1}^{\ngroup} \res_i} - \budget\right) = {\bf{0}}$$

\begin{lemma}
For continuous distributions $\{\candist\}$ with continuous resource allocation, solutions to the above equations are in the form 
$$
\tilde \res_i = \begin{cases}
       0 & F_i(0) > 1 + \tilde \tau \\
       F_i^{-1}(1+\tilde \tau) & F_i(0) \leq 1 + \tilde \tau \
\end{cases}$$
such that 
$$\budget = \sum_{i \in \mathcal{I}} F_i^{-1}(1 + \tilde \tau) \quad  \mathcal{I} = \{i  \ \vert  \ F_i(0) \leq 1 + \tilde \tau \}
$$
No matter what $\{F_i\}$ CDFs and $0\leq \budget < \infty$ values are provided, a solution exists of this form. 
\end{lemma}
\begin{proof}
We will implement this proof by reasoning about the potential solutions to the results above. 
The last equation involves taking the gradient at $\tilde \res$ and produces a vector-valued output. We the derivatives below:
$$\nabla_{\res} \ \tau \left( \p{\sum_{i=1}^{\ngroup} \res_i} - \budget\right) = \tau \cdot {\bf{1}} \text{ which at optimality evaluates to } \tilde \tau \cdot {\bf{1}}$$
$$\sum_{i=1}^{\ngroup}\nabla_{\res} ( \lambda_i(- \res_i)) =  - {\bf{\lambda}} \text{ which evaluates to the vector} - \tilde {\bf{\lambda}}$$
\begin{eqnarray*}\nabla_{\res} \left(\sum_{i=1}^{\ngroup}\mathbb{E}_{\bigcan_i \sim \candist_i}\min(\bigcan_i, \res_i)\right) = \begin{bmatrix} 
1 - F_1(\res_1) \\
\vdots \\ 
1 - F_i(\res_i) \\ 
\vdots \\
1 - F_{\ngroup}(\res_{\ngroup})
\end{bmatrix} \\ 
\text{ which under optimality equals }
\begin{bmatrix} 
1 - F_1(\tilde \res_1) \\
\vdots \\ 
1 - F_i(\tilde \res_i) \\ 
\vdots \\
1 - F_{\ngroup}(\tilde \res_{\ngroup})
\end{bmatrix}
\end{eqnarray*}
Saying that we want this entire value to sum up to 0 means that we require: 
$$(1 - F_i(\tilde \res_i)) - \tilde \lambda_i + \tilde \tau = 0 \quad \Rightarrow \quad  1 - \tilde \lambda_i + \tilde \tau = F_i(\tilde \res_i) \quad \forall i \in [\ngroup]$$
$F_i(\res_i)$, as the CDF of a probability distribution, is always monotonically decreasing in $\res_i$. If $F_i(0) > 1 + \tilde \tau$, then $F_i(\tilde \res_i) > 1 + \tilde \tau$, so in order to satisfy $F_i(\tilde \res_i) = 1 - \tilde \lambda_i + \tilde \tau$, we need $\tilde \lambda_i >0$. By the third requirement ($\res_i \cd \tilde \lambda_i = 0$), this requires that $\tilde \res_i = 0$. This gives us part of the theorem results: $F_i(0) > 1 + \tilde \tau \ \Rightarrow \  \res_i = 0$. 

Suppose instead that $1 + \tilde \tau > F_i(0)$. Because $1 - \tilde \lambda_i + \tilde \tau = F_i(\res_i)$ and $\lambda_i \geq 0$, in order for this equality to hold we must have $F_i(\res_i) > F_i(0)$, which means $\res_i > 0$. Because $\res_i \cd \tilde \lambda_i = 0$, this implies $\tilde \lambda_i = 0$. This tells us that $F_i(\tilde \res_i) = 1 + \tilde \tau$, so $\tilde \res_i = F^{-1}_i(1 + \tilde \tau)$. 

Finally, we consider the case where $F_i(0) = 1 + \tilde \tau$. We remember that $1 - \tilde \lambda_i + \tilde \tau = F_i(\tilde \res_i)$, so $1 + \tilde \tau \geq F_i(\tilde r_i)$, and one of $\tilde \res_i, \tilde \lambda_i$ must be 0. If $\tilde \res_i >0$, then we must have $1 + \tilde \tau = F_i(\tilde \res_i)$, which implies $F_i(\tilde \res_i) = F_i(0)$. Note that this is possible if $F_i$ isn't injective (specifically, if the PDF $f$ is at some points equal to 0). If on the other hand we have $\tilde \res_i = 0$, then $F_i(\res_i) = 1 + \tilde \tau$ because $F_i(\res_i) = F_i(0)$. In either case, $\res_i = F^{-1}_i(1 + \tilde \tau)$. 

We can combine these results to indicate the value of $\tilde \res_i$:
$$
\tilde \res_i = \begin{cases}
       0 & F_i(0) > 1 + \tilde \tau \\
       F_i^{-1}(1+\tilde \tau) & F_i(0) \leq 1 + \tilde \tau \
\end{cases}
$$
Rewriting with $\tau = 1 + \tilde \tau$ and including the equality constraint gives us
$$\budget = \left(\sum_{i=1}^{\ngroup} \tilde \res_i\right)  \quad \Rightarrow \quad \budget = \sum_{i \in \mathcal{I}} F_i^{-1}(\tau) \quad  \mathcal{I} = \{i  \ \mid  \ F_i(0) \leq \tau \}$$

In order to find which $\tau$ satisfies the above requirements, a suggested implementation is to run binary search over $\tilde \tau \in [0, 1]$ to find a value that results in $\sum_{i} \tilde \res_i = \budget$. For $\tau = 0$, $\res_i=0$ satisfies the requirement that $\res_i = F_i^{-1}(0)$ if $f_i(0) = 0$ and $\tilde \res_i = 0$ if $F_i(0) > 0$, producing $\budget = \sum_i \res_i = 0$. For $ \tau = 1$, $\res_i = \infty$ satisfies $\res_i = F^{-1}_i(1)$, producing $\sum_i \res_i = \budget = \infty$. For any $\budget \in (0, \infty)$, $\tau \in [0, 1]$ will suffice.

Finally, if $F_i(\res_i)$ has nonzero support everywhere for $\res_i >0$, the utilization function is strictly concave and so there exists only one location satisfying these requirements. 
\end{proof}

\subsection[Fractional \& discrete]{Fractional allocation over discrete candidate distributions}
The above algorithm provides a method to allocate resources across continuous distributions. This method will not work for discrete candidate distributions because it relies on calculating the derivative of the utilization function, and for discrete allocations, this function is not differentiable. However, there exists a simpler algorithm for allocating resources in this case. Algorithm 1 of \citein{ganchev2009censored} allocates each marginal unit resource so as to maximize $1 - F_i(\res_i)$, and Theorem 1 of the same paper proves that such an algorithm maximizes expected utilization of the resource. For completeness, we include this theorem and proof below, translated into the notation that the rest of this paper uses. However, we wish to emphasize that all of the conent of this proof is taken directly from \citein{ganchev2009censored}. 

\subsubsection{Algorithm and proof from \citein{ganchev2009censored}}
Algorithm \ref{discrete_alloc} describes the discrete allocation algorithm included in \citein{ganchev2009censored}. Note that it differs in its presentation here in a) its notation, which uses the CDF rather than the tail probability, and b) its use of the exact probability distributions rather than estimates. In order to make this adaptation, we used the fact that 
$$T(\res) \coloneqq \sum_{\can = \res}^{\infty}P(\bigcan = \can) = 1 - P(\bigcan \leq \res -1) = 1 - F(\res - 1)$$
\begin{algorithm}
\textbf{Input:} Budget $\budget$, probability distributions $\{P_i\}$\\
\textbf{Output:} An allocation ${\bf{\res}}$ \\
${\bf{\res}} ={\bf{0}}$ \\
\For{$l = 1$ to $\budget$}{
$ j \leftarrow \text{argmax}_i (1 - F_i(\res_i))$ \\ 
$ \res_j \leftarrow \res_j + 1$
}
\caption{Allocation algorithm from \citet{ganchev2009censored}}
\label{discrete_alloc}
\end{algorithm}
\begin{theorem}
Algorithm \ref{discrete_alloc} maximizes the expected utilization of resources allocated over the distributions.  
\end{theorem}
In preparation of proving this theorem, we note that
\begin{eqnarray*}
&& (\res - 1)P(\can = \res - 1) + \res \sum_{\can = \res}^{\infty}P(\bigcan = \can) \\
&=& (\res - 1) \sum_{\can = \res - 1}^{\infty}P(\bigcan = \can) + \sum_{\can = \res}^{\infty}P(\bigcan = \can)
\end{eqnarray*}
The fact can be proved by rearranging the LHS: 
\begin{eqnarray*}(\res - 1)P(\can = \res - 1) + \res \sum_{\can = \res}^{\infty}P(\bigcan = \can)  
\\ = (\res - 1)P(\can = \res - 1) + (\res -1) \sum_{\can = \res}^{\infty}P(\bigcan = \can) + \sum_{\can = \res}^{\infty}P(\bigcan = \can)\\  = (\res - 1) \sum_{\can = \res - 1}^{\infty}P(\bigcan = \can) + \sum_{\can = \res}^{\infty}P(\bigcan = \can)
\end{eqnarray*}

\begin{proof}
[Proof from \citet{ganchev2009censored}]
Because $1 - F_i(\res_i) \geq 1 - F_i(\res_i + 1)$, in greedily selecting to maximize $1 - F_i(\res_i)$, the algorithm returns
$$\text{argmax}_{\bf{\res}}\sum_{i=1}^{\ngroup}\sum_{\can=1}^{\res_i}(1 - F_i(\can)) \text{ s.t. } \sum_{i=1}^{\ngroup}\res_i = \budget$$
We show that the value within the argmax is equivalent to the expected utilization of this allocation. Working with an arbitrary group $i$, 
$$\mathbb{E}_{\bigcan \sim \candist_i}\min(\bigcan, \res_i) = \sum_{\can = 1}^{\res_i - 1}\can \cd P(\bigcan = \can) + \res_i \cd \sum_{\can = \res_i}^{\infty}P(\bigcan = \can) $$
$$= \sum_{\can = 1}^{\res_i - 2}\can \cd P(\bigcan = \can) + (\res_i - 1) \cd P(\bigcan = \res_i - 1) +  \res_i \cd \sum_{\can = \res_i}^{\infty}P(\bigcan = \can) $$
Using the result above, we can rewrite this as: 
$$ = \sum_{\can = 1}^{\res_i - 2}\can \cd P(\bigcan = \can) + (\res_i - 1) \sum_{\can = \res_i -1}^{\infty}P(\bigcan = \can) + \sum_{\can = \res_i}^{\infty} P(\bigcan = \can)$$
If we repeated this procedure once more, we would obtain
\begin{eqnarray*}
&& \sum_{\can = 1}^{\res_i - 3}\can \cd P(\bigcan = \can) + (\res_i - 2) \sum_{\can = \res_i -2}^{\infty}P(\bigcan = \can)\\
&& + \sum_{\can = \res_i-1}^{\infty} P(\bigcan = \can)+  \sum_{\can = \res_i}^{\infty} P(\bigcan = \can)
\end{eqnarray*}
If we repeated this procedure a total of $\res_i-1$ times, we would obtain the result:  $$\mathbb{E}_{\bigcan \sim \candist_i}\min(\bigcan, \res_i) = \sum_{\can = 0}^{\res_i -1}(1 - F_i(\can))$$
This allows us to write 
$$\sum_{i=1}^{\ngroup}\sum_{\can=1}^{\res_i}(1 - F_i(\can)) = \sum_{i=1}^{\ngroup}\sum_{\can=1}^{\res_i}\mathbb{E}_{\bigcan \sim \candist_i}\min(\bigcan, \res_i)$$
Algorithm \ref{discrete_alloc} maximizes the term on the LHS, so it also maximizes the expected utilization. 
\end{proof}

\subsubsection{Extension of Algorithm \ref{discrete_alloc} to fractional resource allocation}
The theorem above assumes integer-valued allocation of resources: each $\res_i$ is increased in integer steps. In this section, we will show that even allowing fractional allocation of resources across discrete distributions, the max-utilizing allocation will still be integer-valued. Given that Algorithm \ref{discrete_alloc} returns an allocation that maximizes utilization over discrete allocations, this same allocation will also maximize utilization over continuous allocations.  
\begin{lemma}
Fractional allocation of resources across discrete distributions produces an allocation ${\bf{\res}}$ with integer values. 
\end{lemma}
\begin{proof}
Suppose by contradiction that the allocation had non-integer values. Because $\sum_i \res_i =\budget \in \mathbb{N}$, there must be multiple groups with $\res_i \not \in \mathbb{N}$. Take $\res_i = n_i + \epsilon_i$ and $\res_j = n_j + \epsilon_j$, with $n_i, n_j \in \mathbb{N}$. The utilization of each group is equal to
\begin{eqnarray*}
&&\sum_{\can = 0}^{n_i} \can \cd f_i(\can) + (n_i + \epsilon_i) \cd (1 - F_i(n_i)) \\ 
&&\sum_{\can = 0}^{n_j} \can \cd f_j(\can) + (n_j + \epsilon_j) \cd (1 - F_j(n_j))
\end{eqnarray*}
WLOG, assume that $1 - F_i(n_i) \geq 1 - F_j(n_j)$.
First, consider the case where $\epsilon_i + \epsilon_j \leq 1$. We will show that we can create an allocation with equal or greater utilization. Take $\res_i' = n_i + \epsilon_i + \epsilon_j$ and $\res_j' = n_j$. This has utilization in each group equal to 
\begin{eqnarray*}
&&\sum_{\can = 0}^{n_i} \can \cd f_i(\can) + (n_i + \epsilon_i + \epsilon_j) \cd (1 - F_i(n_i)) \\
&& \sum_{\can = 0}^{n_j} \can \cd f_j(\can) + n_j \cd (1 - F_j(n_j))
\end{eqnarray*}
Changing from $\{\res_i, \res_j\}$ to $\{\res_i', \res_j'\}$ has a net change on the total utilization across these two groups of 
$$\epsilon_j \cd ((1 - F_i(n_i) - (1 - F_j(n_j))) \geq 0$$
Secondly, we must consider the alternate case where $\epsilon_i + \epsilon_j >1$. Then, take $\res_i' = n_i + 1$ and $\res_j' = n_j + \epsilon_j - (1 - \epsilon_i)$. This has utilization in each group equal to \begin{eqnarray*}
&&\sum_{\can = 0}^{n_i} \can \cd f_i(\can) + (n_i + 1) \cd (1 - F_i(n_i)) \\
&&\sum_{\can = 0}^{n_j} \can \cd f_j(\can) + (n_j + \epsilon_j - (1-\epsilon_i)) \cd (1 - F_j(n_j))
\end{eqnarray*}
Changing from $\{\res_i, \res_j\}$ to $\{\res_i', \res_j'\}$ has a net impact on total utilization of
$$(1 - \epsilon_j) \cd ((1 - F_i(n_i) - (1 - F_j(n_j))) \geq 0$$
In both cases, we have shown that, by turning one of the allocations from a fractional to integer-valued allocation, we can either have the same utilization or improve utilization. By continuing this process, we can turn all of the fractional allocations to discrete ones: for the last pair of groups, $\epsilon_i + \epsilon_j = 1$ because the total amount of resources must be integer-valued. This contradicts the assumption that the fractional allocation was max-utilizing, so the overall max-utilizing allocation must be integer valued. 
\end{proof}
This proof shows that the allocation returned by Algorithm \ref{discrete_alloc} is also max-utilizing under fractional allocation. 

\end{document}